\documentclass[letterpaper, 10pt, conference]{ieeeconf}

\IEEEoverridecommandlockouts                              %

\input{general}

\title{\LARGE \bf
Stability Verification of Neural Network Controllers using Mixed-Integer Programming
}

\author{Roland Schwan$^{1,2}$, Colin N. Jones$^{1}$, and Daniel Kuhn$^{2}$%
\thanks{This work was supported by the Swiss National Science Foundation under the NCCR Automation project, grant agreement 51NF40\_180545.}%
\thanks{We thank Silvia Mastellone for valuable feedback on an earlier version of this paper, and Emilio Maddalena
for providing the experimental setup of the DC-DC power convertor.}%
\thanks{$^{1}$Roland Schwan and Colin N. Jones are with the Automatic Control Lab, EPFL, Switzerland.}%
\thanks{$^{2}$Roland Schwan and Daniel Kuhn are with the Risk Analytics and Optimization Chair, EPFL, Switzerland.}%
\thanks{{\tt\footnotesize \{roland.schwan, colin.jones, daniel.kuhn\}@epfl.ch}}%
}

\begin{document}

\maketitle
\thispagestyle{empty}
\pagestyle{empty}

\begin{abstract}
We propose a framework for the stability verification of Mixed-Integer Linear Programming (MILP) representable control policies. This framework compares a fixed candidate policy, which admits an efficient parameterization and can be evaluated at a low computational cost, against a fixed baseline policy, which is known to be stable but expensive to evaluate. We provide sufficient conditions for the closed-loop stability of the candidate policy in terms of the worst-case approximation error with respect to the baseline policy, and we show that these conditions can be checked by solving a Mixed-Integer Quadratic Program (MIQP). Additionally, we demonstrate that an outer and inner approximation of the stability region of the candidate policy can be computed by solving an MILP. The proposed framework is sufficiently general to accommodate a broad range of candidate policies including ReLU Neural Networks (NNs), optimal solution maps of parametric quadratic programs, and Model Predictive Control (MPC) policies. We also present an open-source toolbox in Python based on the proposed framework, which allows for the easy verification of custom NN architectures and MPC formulations. We showcase the flexibility and reliability of our framework in the context of a DC-DC power converter case study and investigate its computational complexity.
\end{abstract}

\section{Introduction}

MPC has been extremely successful in control applications for refineries and chemical plants \cite{qin2003}, building control \cite{oldewurtel2012}, the control of quadcopters \cite{mueller2013}, robotics \cite{neunert2018}, and power electronics \cite{karamanakos2020a}, \cite{maddalena2021}. The main advantages of MPC are its versatility, stability, and ability to account for input and state constraints. With the transition from the process industry to robotics and power electronics, the sampling times have decreased from hours to only a few milli- or even microseconds. This is especially challenging if one wants to deploy controllers on embedded systems with low computational resources and limited memory.

Hence, ideally, one would like to perform all heavy computations offline and precompute the optimal control law $\psi^\star(\cdot)$ that maps any feasible state to an optimal control input. A well-known technique to compute such an optimal control law~$\psi^\star(\cdot)$ is \textit{explicit} MPC \cite{alessio2009}. For MPC controllers with quadratic cost functions and linear dynamics, $\psi^\star(\cdot)$ is a piecewise affine function defined over a polyhedral partition of the state space. Thus, in contrast to \textit{implicit} MPC, which computes the optimal control input online by solving a different optimization problem for each state, \textit{explicit} MPC precomputes the optimal affine control policy across predefined polytopic regions of the state space. Unfortunately, the required number of polytopes explodes with the dimension of the state and the number of constraints, which render explicit MPC intractable for larger systems. Additionally, the online search for the polytope containing the current state may require excessive processing power or storage space. Although these computational challenges can be mitigated by using search trees \cite{jones2006} or hash tables \cite{bayat2011}, the required memory may remain too large \cite{kvasnica2012}, especially for embedded systems. Recent attempts to approximate explicit MPC directly either suffer from a curse of dimensionality \cite{domahidi2011}, or they rely on expensive set projections \cite{chen2018}, \cite{paulson2020}.

The limited scalability of explicit MPC and its variants has promoted interest in general function approximators of MPC policies, such as deep NNs \cite{psaltis1988}.
Deep NNs are attractive because they can exactly represent predictive controllers for linear systems \cite{karg2020a}, while being relatively inexpensive during online inference \cite{hornik1989}. Unfortunately, one loses the stability guarantees for the learned controllers. Statistical methods are one way to verify the stability of the learned controllers \cite{hertneck2018}, with extensions to filter out severely suboptimal control inputs with high probability \cite{zhang2021}.

However, NNs are not the only viable function approximators. In fact, every continuous nonlinear control law can be represented as the minimizer mapping of a parametric convex program \cite{baes2008}. Even the solutions of parametric linear programs (LPs) can represent any continuous piecewise affine function \cite{hempel2013}. Hence, parametric optimization problems give rise to implicit function approximators, and by leveraging the implicit function theorem, one can directly optimize and learn the underlying problem parameters.
Approaches for calibrating parametric quadratic programs (QPs), general convex programs, and root finding problems are described in \cite{amos2017}, \cite{agrawal2019}, and \cite{bai2019}, respectively.

NN control policies have been successfully employed in applications for controlling chemical plants \cite{kumar2021}, robotic arms \cite{nubert2020} or DC-DC power electronic converters \cite{maddalena2021} at a significant increase in computational speed. Thus, NN-based controllers can not only result in better control performance thanks to a tighter control loop, but also in drastic savings of computational resources.

\subsection{Related Work}

The idea to approximate predictive controllers with NNs enjoys growing popularity. Some of the earliest work dates back to the 90s, when single hidden layer NN approximations were used to learn a nonlinear MPC policy \cite{parisini1995}. More recently, plain vanilla NN architectures have been enhanced with parametric QP implicit layers \cite{maddalena2020}. Closed-loop stability can be guaranteed by projecting the output of the NN into a safety set \cite{paulson2020}. This projection can be viewed as a ``Safety Filter'' \cite{hewing2020}. However, a computationally expensive optimization problem has to be solved online.

The satisfaction of the closed-loop state-input constraints and stability requirements can also be verified by performing a reachability analysis using an MILP representation of the NN controller \cite{karg2020b} or by over approximating the NN output bounds via semidefinite programming formulations \cite{fazlyab2022}.

Other approaches show closed-loop stability by learning and verifying Lyapunov functions \cite{dai2020}, \cite{chen2021}. However, the underlying learner/verifier pattern can be computationally demanding because in each iteration adversarial points have to be found (verifier), and a new candidate Lyapunov function must be constructed, without guarantees of convergence.

Alternatively, one can guarantee stability by using an MILP framework to combine the worst-case approximation error of the NN controller with its Lipschitz constant \cite{fabiani2021}.
In this paper, we propose direct sufficient conditions for closed-loop stability, which can be verified by solving an MIQP. We show empirically that our approach is less conservative but computationally more expensive than that proposed in~\cite{fabiani2021} because MIQPs are generically harder then MILPs.

\begin{diffadd}
Our approach to stability certification is closely related to robustness certification techniques in image classification based on MILPs \cite{bunel2018, tjeng2019} or GPU-accelerated methods \cite{mueller2021}. We refer to these papers for a comprehensive literature review. While we use the same MILP formulations to represent NNs, we extend these MILPs to a broader function classes that are sufficiently expressive to represent parametric QPs, which are needed for the verification of predictive controllers.
\end{diffadd}

\subsection{Contributions}

We introduce a general framework for the verification of NN controllers via mixed-integer programming (MIP) together with an open source toolbox.\footnote{The toolbox can be accessed under the following link:\\\url{https://github.com/PREDICT-EPFL/evanqp}} Specifically, we introduce the concept of MILP-representable verification problems, which is sufficiently general to cater for a variety of commonly used control policies such as ReLU NNs and the optimizer maps of parametric quadratic programs, which include common MPC policies. Given a baseline policy (e.g., an MPC policy) and an approximate policy (e.g., a NN), we propose two  approaches to verify closed-loop stability of the approximate policy:
\begin{itemize}
    \item By computing the worst-case approximation error. Utilizing robust MPC schemes such as Tube MPC \cite{mayne2005}, we can guarantee closed-loop stability by constraining the approximation error to fall within the disturbance set of the robust MPC scheme.
    \item By providing sufficient conditions for closed-loop stability, which can be directly formulated and checked by solving an MIQP. We show that our conditions are less conservative in practice than existing methods, i.e., they are able to verify the closed-loop stability of a larger class of approximate policies. In addition, we show that outer and inner approximations of the stability region of the approximate control law can be found by solving MILPs.
\end{itemize}
We then exemplify the application of the proposed method on a case-study of a DC-DC power convertor using different NN architectures and MPC formulations. We compare our methods against state-of-the-art approaches and show the superiority of our formulation. Additionally, we provide numerical experiments to show the performance of our approach.

The toolbox is written in Python and allows the automatic transcription of the discussed verification problems as MIPs. It can directly import NN architectures from Pytorch \cite{paszke2019} and parametric QPs from CVXPY \cite{diamond2016}. This not only allows for a simplified problem formulation but also a tight integration with the existing machine learning ecosystem.

\section*{Notation}
We denote the set of real numbers by $\mathbb{R}$, the set of $n$-dimensional real-valued vectors by $\mathbb{R}^n$ and the set of $n \times m$-dimensional real-valued matrices by $\mathbb{R}^{n \times m}$. Furthermore, we denote the subspace of symmetric matrices in $\mathbb{R}^{n \times n}$ by $\mathbb{S}^n$ and the cone of positive semi-definite and definite matrices by $\mathbb{S}^n_+$ and $\mathbb{S}^n_{++}$, respectively. We use $I_n$ to denote the $n$-dimensional identity matrix, $\boldsymbol{1}_n$ to denote the $n$-dimensional column vector of ones, and $\diag(\cdot)$ to represent the mapping that transforms a column vector to the corresponding diagonal matrix. Given two sets $A,B\subseteq \mathbb R^n$, we denote their Minkowski sum as $A \oplus B \coloneqq \left\{ a + b \;\middle|\; a \in A, b \in B \right\}$ and their Pontryagin difference as $A \ominus B \coloneqq \left\{ a \;\middle|\; a \oplus B \subseteq A \right\}$. The interior of a set $S\subseteq \mathbb R^n$ is denoted by $\inter(S)$. \begin{diffadd}Finally, we use $\gr(f)=\left\{(x,y) \;\middle|\; y = f(x)\right\}$ to denote the graph of a function $f$.\end{diffadd}

\section{MILP-Representable Verification Problems}

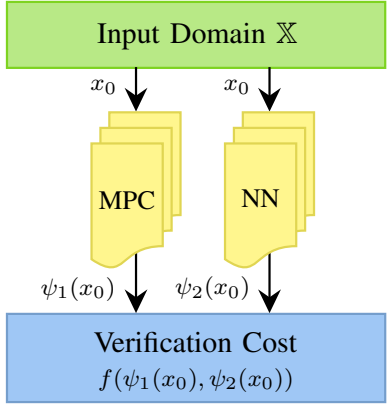
\begin{figure}[ht]
\centering
\resizebox{0.3\textwidth}{!}{\tikzset{every picture/.style={line width=0.75pt}} %

\begin{tikzpicture}[x=0.75pt,y=0.75pt,yscale=-1,xscale=1]

\draw  [color={rgb, 255:red, 126; green, 211; blue, 33 }  ,draw opacity=1 ][fill={rgb, 255:red, 184; green, 233; blue, 134 }  ,fill opacity=1 ] (17,10) -- (187,10) -- (187,39.67) -- (17,39.67) -- cycle ;
\draw  [color={rgb, 255:red, 117; green, 158; blue, 204 }  ,draw opacity=1 ][fill={rgb, 255:red, 159; green, 199; blue, 245 }  ,fill opacity=1 ] (17,150.33) -- (187,150.33) -- (187,190) -- (17,190) -- cycle ;
\draw    (75,39.33) -- (75,56.33) ;
\draw [shift={(75,59.33)}, rotate = 270] [fill={rgb, 255:red, 0; green, 0; blue, 0 }  ][line width=0.08]  [draw opacity=0] (10.72,-5.15) -- (0,0) -- (10.72,5.15) -- (7.12,0) -- cycle    ;
\draw    (135,39.33) -- (135,56.33) ;
\draw [shift={(135,59.33)}, rotate = 270] [fill={rgb, 255:red, 0; green, 0; blue, 0 }  ][line width=0.08]  [draw opacity=0] (10.72,-5.15) -- (0,0) -- (10.72,5.15) -- (7.12,0) -- cycle    ;
\draw    (75,119.67) -- (75,147) ;
\draw [shift={(75,150)}, rotate = 270] [fill={rgb, 255:red, 0; green, 0; blue, 0 }  ][line width=0.08]  [draw opacity=0] (10.72,-5.15) -- (0,0) -- (10.72,5.15) -- (7.12,0) -- cycle    ;
\draw    (135,119.67) -- (135,147) ;
\draw [shift={(135,150)}, rotate = 270] [fill={rgb, 255:red, 0; green, 0; blue, 0 }  ][line width=0.08]  [draw opacity=0] (10.72,-5.15) -- (0,0) -- (10.72,5.15) -- (7.12,0) -- cycle    ;
\draw  [color={rgb, 255:red, 224; green, 213; blue, 82 }  ,draw opacity=1 ][fill={rgb, 255:red, 255; green, 246; blue, 142 }  ,fill opacity=1 ] (63,59.67) -- (95,59.67) -- (95,105.87) .. controls (75,105.87) and (79,122.53) .. (63,111.75) -- cycle ; \draw  [color={rgb, 255:red, 224; green, 213; blue, 82 }  ,draw opacity=1 ][fill={rgb, 255:red, 255; green, 246; blue, 142 }  ,fill opacity=1 ] (59,66.67) -- (91,66.67) -- (91,112.87) .. controls (71,112.87) and (75,129.53) .. (59,118.75) -- cycle ; \draw  [color={rgb, 255:red, 224; green, 213; blue, 82 }  ,draw opacity=1 ][fill={rgb, 255:red, 255; green, 246; blue, 142 }  ,fill opacity=1 ] (55,73.67) -- (87,73.67) -- (87,119.87) .. controls (67,119.87) and (71,136.53) .. (55,125.75) -- cycle ;
\draw  [color={rgb, 255:red, 224; green, 213; blue, 82 }  ,draw opacity=1 ][fill={rgb, 255:red, 255; green, 246; blue, 142 }  ,fill opacity=1 ] (123,59.67) -- (155,59.67) -- (155,105.87) .. controls (135,105.87) and (139,122.53) .. (123,111.75) -- cycle ; \draw  [color={rgb, 255:red, 224; green, 213; blue, 82 }  ,draw opacity=1 ][fill={rgb, 255:red, 255; green, 246; blue, 142 }  ,fill opacity=1 ] (119,66.67) -- (151,66.67) -- (151,112.87) .. controls (131,112.87) and (135,129.53) .. (119,118.75) -- cycle ; \draw  [color={rgb, 255:red, 224; green, 213; blue, 82 }  ,draw opacity=1 ][fill={rgb, 255:red, 255; green, 246; blue, 142 }  ,fill opacity=1 ] (115,73.67) -- (147,73.67) -- (147,119.87) .. controls (127,119.87) and (131,136.53) .. (115,125.75) -- cycle ;

\draw (102.61,25.17) node   [align=left] {Input Domain $\displaystyle \mathbb{X}$};
\draw (101.69,162.83) node   [align=left] {Verification Cost};
\draw (60.5,48.17) node  [font=\footnotesize]  {$x_{0}$};
\draw (120.5,48.17) node  [font=\footnotesize]  {$x_{0}$};
\draw (49,139) node  [font=\footnotesize]  {$\psi _{1}( x_{0})$};
\draw (109.5,138.83) node  [font=\footnotesize]  {$\psi _{2}( x_{0})$};
\draw (71,99.17) node  [font=\small] [align=left] {MPC};
\draw (131,98.17) node  [font=\small] [align=left] {NN};
\draw (102,180.17) node  [font=\footnotesize]  {$f( \psi _{1}( x_{0}) ,\psi _{2}( x_{0}))$};

\end{tikzpicture}}
\caption{Schematic of the verification approach.}
\label{fig:overview}
\end{figure}

A verification problem consists of a baseline control policy $\psi_1(x)$ (e.g., an MPC policy) and an approximate policy $\psi_2(x)$ (e.g., a NN) on a common input domain $\mathbb{X}$. The outputs of the policies then enter a verification cost function $f(\psi_1(x),\psi_2(x))$ that represents the objective of the verification problem. An overview of the proposed architecture is shown in Figure~\ref{fig:overview}.

\begin{diffadd}
In the following, we define the concept of MILP-representable functions. A MILP-representable function can be represented exactly by linear equality and inequalities with continuous and binary decision variables.
\end{diffadd}

\begin{definition} \label{def:MILP_rep}
A function $\psi: \mathbb{X} \rightarrow \mathbb{U}$ with domain $\mathbb{X} \subseteq \mathbb{R}^{n}$ and range $\mathbb{U} \subseteq \mathbb{R}^{m}$ is MILP-representable if there exists a polyhedral set $P \subseteq \mathbb{R}^{n} \times \mathbb{R}^{m} \times \mathbb{R}^{c} \times \mathbb{R}^{b}$ such that $(x,u)\in \gr(\psi)$ if and only if there exists $z \in \mathbb{R}^c$ and $\beta \in \{0,1\}^b$ such that $(x,u,z,\beta) \in P$.
\end{definition}

A wide range of functions are MILP-representable. In fact, one can show that the MILP-representable functions are dense in the family of continuous functions on a compact domain $\mathbb{X}$ with respect to the supremum norm. Examples of MILP-representable functions include the optimal solution maps of parametric QPs (e.g., MPC policies), ReLU NNs, and piecewise affine functions. Additionally, compositions of MILP-representable functions are also MILP-representable.
\begin{lemma} \label{lem:MILP_comp}
If two functions $\psi_i: \mathbb{X}_i \rightarrow \mathbb{U}_i, \; i=1,2$ with $\mathbb{X}_2 \subseteq \mathbb{U}_1$ are both MILP-representable, then the composition $\psi_1 \circ \psi_2: \mathbb{X}_1 \rightarrow \mathbb{U}_2$ is MILP-representable.
\end{lemma}
\begin{proof}
Let $P_1 \subseteq \mathbb{R}^{n} \times \mathbb{R}^{m} \times \mathbb{R}^{c_1} \times \mathbb{R}^{b_1}$ and $P_2 \subseteq \mathbb{R}^{m} \times \mathbb{R}^{o} \times \mathbb{R}^{c_2} \times \mathbb{R}^{b_2}$ be the polyhedral sets corresponding to the MILP representations of $\psi_1$ and $\psi_2$ respectively. Then, $(x_i,u_i) \in \gr(\psi_i)$ if and only if there exist $z_i \in \mathbb{R}^{c_i}$ and $\beta_i \in \{0,1\}^{b_i}$ such that $(x_i,u_i,z_i,\beta_i) \in P_i$ for $i=1,2$. To compute the composition of $\psi_1$ and $\psi_2$, we now set the output $u_1$ of $\psi_1$ equal to the input $x_2$ of $\psi_2$. Specifically, we define new variables $z=(z_1,z_2,u_1)$ and $\beta=(\beta_1,\beta_2)$, \begin{diffadd}and we construct a new polyhedral set $P \subseteq \mathbb{R}^{n} \times \mathbb{R}^{o} \times \mathbb{R}^{c_1+c_2+m} \times \mathbb{R}^{b_1+b_2}$, which includes all linear constraints defining the original polyhedra~$P_1$ and$P_2$ as well as the equality constraint $u_1=x_2$.\end{diffadd} %
Hence, by construction, it holds that $(x_1,u_2)\in \gr(\psi_1 \circ \psi_2)$ if and only if $(x_1,u_2,z,\beta)\in P$, showing that the composition $\psi_1 \circ \psi_2$ is indeed MILP-representable.
\end{proof}

In order to formalize the verification problem for comparing the approximate control policy $\psi_2(\cdot)$ against the baseline policy $\psi_1(\cdot)$, define $\mathbb{X}$ as the relevant input domain of the control polices. For any property to be verified, we can pose the verification problem as the task of certifying the non-negativity of a function $f(\cdot)$ over the input domain $\mathbb{X}$.
\begin{equation} \label{eq:veri_ineq}
    0 \leq \min_{x_{0} \in \mathbb{X}} f(\psi_1(x_{0}) \psi_2(x_{0})).
\end{equation}

\begin{proposition} \label{lem:veri_MILP_formulation}
If $\mathbb{X}$ is a polyhedron and if $\psi_1$, $\psi_2$, and $f$ are MILP-representable, then the verification problem on the right-hand side of \eqref{eq:veri_ineq} can be posed as the following optimization problem, which is equivalent to an MILP:
\begin{equation} \label{eq:veri_milp}
\begin{aligned}
\min_{\tau,x_0,u_1,u_2}  \quad & \tau \\
\text { s.t. } \quad\;\;\: & x_{0} \in \mathbb{X} \\
& (x_0,u_1) \in \gr(\psi_1) \\ 
& (x_0,u_2) \in \gr(\psi_2) \\
& (u_1,u_2,\tau) \in \gr(f).
\end{aligned}
\end{equation}
\end{proposition}\medskip
\begin{proof}
The equivalence of the optimization problem in \eqref{eq:veri_ineq} and \eqref{eq:veri_milp} is immediate. In addition, one readily verifies that problem \eqref{eq:veri_milp} is equivalent to an MILP because $\mathbb{X}$ is a polyhedron and $f$, $\psi_1$ and $\psi_2$ are MILP-representable.
\end{proof}

In the next sections, we consider specific MILP-representable functions relevant for control applications.

\subsection{ReLU Neural Networks} \label{subsec:relu_nn}

We now demonstrate that NNs with ReLU activation functions are MILP-representable. A fully-connected feed-forward NN $\psi: \mathbb{X} \subseteq \mathbb{R}^n \rightarrow \mathbb{U} \subseteq \mathbb{R}^m$ is a function representable as $\psi = \psi_\ell \circ \psi_{\ell-1} \circ \dots \circ \psi_1$, where $\psi_i: \mathbb{X}_i \subseteq \mathbb{R}^{n_{i-1}} \rightarrow \mathbb{U}_i \subseteq \mathbb{R}^{n_i}$ corresponding to the $i$-th layer is defined through the componentwise maximum
\begin{equation*}
    \psi_i(z_{i-1})=\relu(z_{i-1})=\max(0, W_iz_{i-1}+b_i)
\end{equation*}
for some weight matrix $W_i \in \mathbb{R}^{n_i \times n_{i-1}}$ and bias vector $b_i \in \mathbb{R}^{n_i}$, $i=1,\dots,\ell$, and where $n_0=n$ is the input dimension and $n_\ell=m$ is the output dimension of the NN.

To prove that $\psi$ is MILP-representable, we use the piecewise linearity of the ReLU activation functions to represent them via linear constraints and binary decision variables indicating which piece is active \cite{bunel2018, tjeng2019}.

\begin{lemma} \label{lem:NN}
If $\mathbb{X}$ is a compact polyhedron, then the ReLU NN $\psi: \mathbb{X} \rightarrow \mathbb{U}$ is MILP-representable, and there exist constants $\underline{m}_i$, $\overline{m}_i$ such that
\begin{equation} \label{eq:ffnn_milp}
    \gr(\psi) = 
    \left\{
    (x,u) \;\middle|\; \begin{aligned}
    \forall i &=1,\dots,n-1, \\
    \exists \beta_i & \in \{0,1\}^{n_i}, \\
    \exists z_i & \in \mathbb{R}^{n_i}: \\
    z_0 &= x, \\
    z_0 &\in \mathbb{X}, \\
    z_i &\geq 0\\
    z_i &\geq W_iz_{i-1} + b_i, \\
    z_i &\leq W_iz_{i-1} + b_i \\
        & \quad - \diag(\overline{m}_i)(\boldsymbol{1}-\beta_i),\\
    z_i &\leq \diag(\underline{m}_i)\beta_i,\\
    u &= W_nz_{n-1}+b_n
    \end{aligned}
    \right\}.
\end{equation}
\end{lemma} \medskip
\begin{proof}
As $\mathbb{X}$ is bounded, there exist $\underline{x}, \overline{x} \in \mathbb{R}^n$ such that $\mathbb{X} \subseteq [\underline{x}, \overline{x}]$. Assume first that $\psi: \mathbb{X} \rightarrow \mathbb{U}$ is a single ReLU neuron defined through $u = \max(0, x)$. Then the graph of $\psi$ can be represented via mixed-integer constraints as follows.
\begin{equation} \label{eq:relu_graph}
    \gr(\psi) = 
    \left\{
    (x,u) \;\middle|\; \begin{aligned}
    \exists \beta &\in \{0,1\}^n: \\
    x &\in \mathbb{X}, \\
    u &\geq 0, \\
    u &\geq x, \\
    u &\leq x - \diag(\underline{x})(1-\beta), \\
    u &\leq \diag(\overline{x})\beta
    \end{aligned}
    \right\}
\end{equation}
Here, the binary decision variable $\beta$ indicates whether the neuron is active $(u=x)$ or inactive $(u=0)$. Extending this reasoning to the full NN and given the element-wise upper and lower bounds $\underline{m}_i \leq W_iz_{i-1}+b_i \leq \overline{m}_i$ on the input of the $i$-th activation layer, we can use Lemma~\ref{lem:MILP_comp} to derive the graph representation \eqref{eq:ffnn_milp}, where the binary decision variables $\beta_i$ correspond to layer $i$. The constants $\underline{m}_i$ and $\overline{m}_i$ always exist because $\mathbb{X}$ is compact, $\psi_i$ is continuous for every $i=1,\dots,\ell$, and continuous images of compact sets are compact.
\end{proof}

\begin{remark}
The bounds $\underline{m}_i$ and $\overline{m}_i$ can be computed using interval arithmetic \cite{wong2018}, zonotope propagation \cite{gehr2018}, or linear programming \cite{tjeng2019}. State-of-the-art MILP solvers such as Gurobi \cite{gurobi} are based on branch-and-cut methods. Hence, tightening the bounds $\underline{m}_i$ and $\overline{m}_i$ can reduce computation times in practice since branches can be pruned more efficiently. If $\underline{m}_i$ and $\overline{m}_i$ can not be calculated, e.g., if $\mathbb{X}$ is unbounded, then $\underline{m}_i$ and $\overline{m}_i$ correspond to the usual ``big-M" constants and have to be determined endogenously.
\end{remark}

\begin{remark}
Lemma~\ref{lem:NN} extends to NNs with general piecewise linear activation functions such as leaky ReLUs.
Note also that the convex hull of $\gr(\psi)$ is a strict subset of the polyhedron obtained by relaxing $\beta_i \in [0,1]$ in \eqref{eq:ffnn_milp}. The approximation quality can be improved by adding valid cuts~\cite{anderson2020}.
\end{remark}

\subsection{Box Saturation}

Control policies are often subject to physical constraints. To enforce these constraints, one may project the output of an approximate policy to the feasible set. For example, the projection $\psi$ onto the box $[\underline{x}, \overline{x}]\subseteq \mathbb{R}^m$ with $\underline{x} \leq \overline{x}$ can be viewed as a vector of element-wise projections of the form
\begin{equation} \label{eq:box_saturation}
\psi(x)_i =\left\{\begin{array}{l}
\underline{x}_i \text { if } u_i < \underline{x}_i, \\
u_i \text { if } \underline{x}_i \leq u_i \leq \overline{u}_i, \\
\overline{x}_i \text { if } u_i > \overline{x}.
\end{array}\right.
\end{equation}
\begin{lemma}
The projection $\psi$ onto the box $[\underline{x}, \overline{x}] \subseteq \mathbb{R}^m$ is MILP-representable.
\end{lemma}
\begin{proof}
The projection $\psi$ defined through \eqref{eq:box_saturation} can be rewritten in the following form using the ReLU function:
\begin{equation*}
    \psi(x) = \overline{x} - \relu(\overline{x} - (\relu(x - \underline{x}) + \underline{x})).
\end{equation*}
Hence, $\psi$ can be recognized as a special instance of a ReLU NN, which is MILP-representable according to Lemma~\ref{lem:NN}.
\end{proof}

\subsection{Parametric QPs} \label{subsec:parametric_qps}

We now show that the optimal solution mappings of parametric QPs are MILP-representable. Specifically, we define $\psi(x)$ as the unique minimizer of the parametric QP
\begin{equation} \label{eq:pQP}
\begin{aligned}
\min_{z} \quad & \frac{1}{2} z^{T} P z + (Qx + q)^Tz \\
\text {s.t.}\, \quad & Az=Bx+b \\
& Fz \leq Gx + g,
\end{aligned}
\end{equation}
with internal decision variable $z \in \mathbb{R}^{n_z}$ and matrices $P \in \mathbb{S}^{n_z}_{++}$, $Q \in \mathbb{R}^{n_z \times n}$, $q \in \mathbb{R}^{n_z}$, $A \in \mathbb{R}^{n_\text{eq} \times n_z}$, $B \in \mathbb{R}^{n_\text{eq} \times n}$, $b \in \mathbb{R}^{n_\text{eq}}$, $F \in \mathbb{R}^{n_\text{ineq} \times n_z}$, $G \in \mathbb{R}^{n_\text{ineq} \times n}$, and $g \in \mathbb{R}^{n_\text{ineq}}$.

\begin{lemma} \label{lem:qp_MILP}
If the parametric QP \eqref{eq:pQP} is feasible for every $x\in\mathbb{X}$, then its optimal solution mapping $\psi(x)$ is MILP-representable.
\end{lemma}
\begin{proof}
Since the feasible set of problem \eqref{eq:pQP} is a polyhedron, it admits a Slater point whenever it is non-empty. Since the objective function of \eqref{eq:pQP} is also strictly convex, $z^\star(x)$ for $x\in\mathbb{X}$ is uniquely determined by the necessary and sufficient Karush-Kuhn-Tucker (KKT) conditions \cite{bertsekas99}
\begin{align*}
    & \begin{array}{l}
    \text {Primal feasibility}\\
    \left\lfloor Az = Bx + b, \quad Fz \leq Gx + g, \right.
    \end{array} \\
    & \begin{array}{l}
    \text {Stationarity}\\
    \left\lfloor Pz + (Qx + q) + A^T\mu + F^T\lambda = 0, \right.
    \end{array} \\
    & \begin{array}{l}
    \text {Dual feasibility}\\
    \left\lfloor \lambda \geq 0, \right.
    \end{array} \\
    & \begin{array}{l}
    \text {Complementarity}\\
    \left\lfloor \lambda^T(Fz-Gx-g) = 0, \right.
    \end{array}
\end{align*}
where $\mu \in \mathbb{R}^{n_\text{eq}}$ and $\lambda \in \mathbb{R}^{n_\text{ineq}}$ are the Lagrange multipliers associated with the equality and inequality constraints, respectively. By introducing a binary decision variable $\beta_i$ that evaluates to $0$ if $F_iz < g_i$ and to $1$ if $\lambda_i > 0$ for each $i=1,\dots,n_\text{ineq}$, we can linearize the complementarity condition as
\begin{equation} \label{eq:KKT_complementarity_big_M}
\left.\begin{array}{l}
0 \leq \lambda_i \leq M\beta_i, \\
0 \leq g_i - F_iz \leq M(1-\beta_i)
\end{array}\right\} \forall i=1,\dots,n_\text{ineq},
\end{equation}
where $M$ is a suitable ``big-M" constant \cite{bemporad1999}. Therefore, the graph of $\psi$ can be represented as
\begin{equation*}
    \gr(\psi) = 
    \left\{
    (x,z) \;\middle|\; \begin{aligned}
    \exists \beta &\in \{0,1\}^{n_\text{ineq}}, \\
    \exists z &\in \mathbb{R}^{n_z}, \\
    \exists \lambda &\in \mathbb{R}^{n_\text{ineq}}, \\
    \exists \mu &\in \mathbb{R}^{n_\text{eq}}: \\
    x &\in \mathbb{X}, \\
    0 &= Az - Bx - b, \\
    0 &= Pz + (Qx + q) + A^T\mu \\
    & \quad + F^T\lambda, \\
    0 &\leq \lambda_i \leq M\beta_i~\forall i=1,\dots,n_\text{ineq}, \\
    0 &\leq g_i + (Gx)_i - F_iz \\
      &\leq M(1-\beta_i)~\forall i=1,\dots,n_\text{ineq}
    \end{aligned}
    \right\}.
\end{equation*}
Thus, the claim follows.
\end{proof}

\begin{remark}
Certain MILP solvers accept logical constraints and automatically transform them into ``big-M" constraints or special-ordered set constraints. In these cases, the MILP representations of logical constraints are optimally chosen by the solver, which leads to faster convergence. Tight ``big-M" bounds suitable for numerical purposes can also be calculated by solving auxiliary LPs. Details are relegated to Appendix~\ref{sec:bound_improvementes}.
\end{remark}

\subsection{Piecewise Affine Functions over Polyhedral Sets} \label{subsec:pwa}

Piecewise affine functions defined over polyhedral sets lend themselves for modeling piecewise affine hybrid dynamical systems \cite{marcucci2019}.
A continuous function $\psi: \mathbb{X} \rightarrow \mathbb{U}$ is called piecewise affine if there exist polyhedra $\mathbb{X}_i=\left\{x \in \mathbb{X} \;\middle|\; F_i x \leq g_i \right\}$, $i\in \mathcal{I}$, with $\mathbb{X}=\bigcup_{i\in\mathcal{I}} \mathbb{X}_i$ as well as matrices $A_i\in\mathbb{R}^{m \times n}$ and vectors $c_i \in \mathbb{R}^n$, $i\in \mathcal{I}$, such that
\begin{equation} \label{eq:pwa}
    \psi(x) = A_ix + c_i \quad \forall x\in\mathbb{X}_i, \; \forall i\in\mathcal{I},
\end{equation}
where $\mathcal{I}$ is a finite index set.

\begin{lemma}
If $\mathbb{X}$ is a compact polyhedron, then the piecewise affine function defined in \eqref{eq:pwa} is MILP-representable.
\end{lemma}
\begin{proof}
By \cite[Theorem~3.5]{marcucci2019}, \begin{diffadd}the piecewise affine function \eqref{eq:pwa} can be exactly represented using MILP constraints, i.e.,\end{diffadd} there exist binary variables $\beta_i\in\{0,1\}$, $i\in\mathcal{I}$, such that
\begin{subequations}
\begin{align}
    &F_ix_i \leq \beta_ig_i \quad \forall i \in \mathcal{I}, \label{eq:piecewise_a}\\
    &1 = \sum_{i\in\mathcal{I}} \beta_i, \; x = \sum_{i\in\mathcal{I}} x_i, \\
    &\psi(x) = \sum_{i\in\mathcal{I}} (A_ix_i+\beta_ic_i).
\end{align}
\end{subequations}
Note first that $\mathbb{X}_i \subseteq \mathbb{X}$ inherits compactness from $\mathbb{X}$. The constraint \eqref{eq:piecewise_a} implies that $x_i=0$ whenever $\beta_i=0$. Indeed, if $\beta_i=0$, then \eqref{eq:piecewise_a} forces $x_i$ to fall within the recession cone of the compact polyhedron $\mathbb{X}_i$, which coincides with the singleton $\{0\}$. If $\beta_i=1$, on the other hand, then \eqref{eq:piecewise_a} forces $x_i$ to fall within $\mathbb{X}_i$.
The graph of $\psi$ thus admits the following MILP representation.
\begin{equation} \label{eq:pwa_graph}
    \gr(\psi) = 
    \left\{
    (x,u) \;\middle|\; \begin{aligned}
    \exists \beta_i &\in \{0,1\} \quad \forall i \in \mathcal{I}, \\
    \exists x_i &\in \mathbb{R}^n \quad\quad\, \forall i \in \mathcal{I}: \\
    F_ix_i &\leq \beta_ig_i \quad\;\,\, \forall i \in \mathcal{I}, \\
    1 &= \sum_{i\in\mathcal{I}} \beta_i, \\
    x &= \sum_{i\in\mathcal{I}} x_i, \\
    u &= \sum_{i\in\mathcal{I}} (A_ix_i+\beta_ic_i)
    \end{aligned}
    \right\}
\end{equation}
Hence, the claim follows.
\end{proof}

\section{Computation of the Approximation Error} \label{sec:error_computation}

Consider now a verification problem of the form \eqref{eq:veri_ineq} with baseline policy $\psi_1(\cdot)$ and approximate policy $\psi_2(\cdot)$, and set the error function $f(\cdot)$ to the infinity norm $\|\cdot\|_\infty$. The worst-case error over a bounded polytopic input set $\mathbb{X}$ is given by
\begin{equation} \label{eq:abs_error_orig}
    \gamma = \max_{x \in \mathbb{X}} \|\psi_1(x) - \psi_2(x)\|_\infty.
\end{equation}
The infinity norm is a natural choice for measuring the mismatch between $\psi_1$ and $\psi_2$, and it is MILP-representable.

\begin{lemma} \label{lem:inf_norm_MILP}
The infinity norm $f(t)=\|t\|_\infty$ is MILP-representable on any bounded polytopic set $\mathbb{X} \subseteq \mathbb{R}^m$.
\end{lemma}

Instead of the infinity norm, one could also use the 1-norm to quantify the worst-case error in \eqref{eq:abs_error_orig}.
\begin{corollary} \label{cor:1_norm}
The 1-norm $f(t)=\|t\|_1$ is MILP-repre\-sen\-table on any bounded polytopic set $\mathbb{X} \subseteq \mathbb{R}^m$.
\end{corollary}

The proofs of Lemma~\ref{lem:inf_norm_MILP} and Corollary~\ref{cor:1_norm} are standard and therefore relegated to Appendices~\ref{subsec:proof_lem:inf_norm_MILP} and~\ref{subsec:proof_cor:1_norm}, respectively. Lemma~\ref{lem:inf_norm_MILP} implies via Proposition~\ref{lem:veri_MILP_formulation} that if $\psi_1(x)$ and $\psi_2(x)$ are MILP-representable, then the worst-case approximation error \eqref{eq:abs_error_orig} can be computed by solving an MILP.

\begin{remark}
By solving $2m$ MILPs one can also construct a tight bounding box $[\underline{\gamma}, \overline{\gamma}]\subseteq\mathbb{R}^m$ that covers the approximation error $\psi_1(x) - \psi_2(x)$ for every $x\in\mathbb{X}$. Indeed, to compute the components of the vector $\underline{\gamma}$, we solve $m$ MILPs of the form
\begin{equation} \label{eq:error_box}
\begin{aligned}
\underline{\gamma}_i = \min_{x,u_1,u_2} \quad & e_i^T(u_1 - u_2) \\
\text{s.t.} ~~\quad & x \in \mathbb{X} \\
& (x,u_1) \in \gr(\psi_1) \\
& (x,u_2) \in \gr(\psi_2), \\
\end{aligned}
\end{equation}
where $e_i\in \mathbb{R}^{m}$ is the $i$-th vector of the canonical basis. To compute the components of $\overline{\gamma}$, we simply convert the minimization operator in~\eqref{eq:error_box} to a maximization operator.
\end{remark}

\subsection{Stability Verification Against Robust MPC} \label{subsec:_verification_robust}

 \begin{diffadd}
 We can use the worst-case approximation error \eqref{eq:abs_error_orig} to verify an approximate MPC scheme against a robustly stable MPC policy. As an example, consider the linear dynamical system
\begin{equation}\label{eq:linear_system}
    x^+ = Ax+Bu
\end{equation}
with state $x \in \mathbb{R}^n$, input $u \in \mathbb{R}^m$ and system matrices $A\in\mathbb R^{n\times n}$ and $B\in\mathbb R^{n\times m}$. Let $\psi_1(x)$ be a control policy that is robust against additive input disturbances in the set $\mathbb{W} = \left\{ w \in \mathbb{R}^m \;\middle|\; \|w\|_\infty \leq \hat{\gamma} \right\}$.
Such a control policy can be obtained via the Tube MPC approach \cite{mayne2005}, for instance. Additionally, suppose that there exists a feedback gain matrix $K\in\mathbb{R}^{m\times n}$ such that $A_K = A + BK$ is asymptotically stable, and let $\mathcal{E}\subseteq \mathbb R^n$ be the disturbance invariant set for the controlled uncertain system $x^+=A_Kx + Bw$, which satisfies
\begin{equation*}
    A_K \mathcal{E} \oplus B \mathbb{W} \subseteq \mathcal{E}.
\end{equation*}
In the context of Tube MPC, $K$ corresponds to the solution of a discrete LQR problem, and $\mathcal{E}$ is the minimum robust invariant set with respect to the feedback gain matrix $K$, to which the Tube MPC policy~$\psi_1(\cdot)$ will converge in closed loop. In Section~\ref{subsec:example_robust_mpc} we will see that~$\psi_1(\cdot)$ is MILP-representable and that the invariant set $\mathcal{E}$ can be efficiently computed offline.

Let now $\tilde{\psi}_2(\cdot)$ be an approximation of $\psi_1(\cdot)$. For example, $\psi_2(\cdot)$ can be obtained by sampling the baseline policy $\psi_1(\cdot)$ at different points of the state space and learning a compatible input-output map by training a NN on the samples. Then, set
\begin{equation*}
    \psi_2(x) = \left\{\begin{array}{ll}
Kx & \text { if } x \in \mathcal{E}, \\
\tilde{\psi}_2(x) & \text { otherwise}.
\end{array}\right.
\end{equation*}
The following theorem provides a means to verify the closed-loop asymptotic stability of approximate policies like $\psi_2(\cdot)$.
\begin{theorem}
The control policy $\psi_2(\cdot)$ constructed from $\tilde \psi_2(\cdot)$, the feedback gain matrix $K$ and the invariant set $\mathcal E$ is asymptotically stable in closed loop on a polyhedron $\mathbb{X}$ if
\begin{equation} \label{eq:error_robust_MPC_verification}
\begin{aligned}
0 \leq \min_{\tau,x,u_1,u_2} \quad & \hat{\gamma} - \tau \\
\text{s.t.} \qquad & x \in \mathbb{X} \\
& (x,u_1) \in \gr(\psi_1) \\
& (x,u_2) \in \gr(\tilde{\psi}_2) \\
& (u_1 - u_2,\tau) \in \gr(\|\cdot\|_\infty),
\end{aligned}
\end{equation}
and if the robust control policy $\psi_1(\cdot)$ converges asymptotically to the set $\mathcal{E}$. The minimization problem in~\eqref{eq:error_robust_MPC_verification} is equivalent to a MILP if $\psi_1(x)$ and $\tilde{\psi}_2(x)$ are MILP-representable.

\end{theorem}\medskip
\begin{proof}
Let $\gamma$ be the worst case approximation error between $\psi_1(x)$ and $\tilde{\psi}_2(x)$ over the domain $\mathbb{X}$ as defined in~\eqref{eq:abs_error_orig}. Hence, $\tilde{\psi}_2(x)$ applied to the real system deviates from $\psi_1(x)$ at most by $\gamma$ with respect to the infinity norm. Since the original policy $\psi_1(x)$ is robust against input disturbances of magnitude up to $\hat{\gamma}$ with respect to the infinity norm, the inequality $\gamma \leq \hat{\gamma}$ is sufficient for $\tilde{\psi}_2(x)$ to be stabilizing, and the state $x$ converges asymptotically to the set $\mathcal{E}$. As soon as the set $\mathcal{E}$ is reached, the controller $\psi_2(x)$ switches from $\tilde\psi_2(x)$ to the linear control law $Kx$. Since $K$ asymptotically stabilizes the underlying system and $\mathcal{E}$ is an invariant set under the linear policy $Kx$, it follows that~$\psi_2(x)$ is asymptotically stable. We can verify the inequality $\gamma \leq \hat{\gamma}$ by solving \eqref{eq:error_robust_MPC_verification} where, at optimality, $\tau$ coincides with the worst case approximation error~$\gamma$. If $\psi_1(x)$ and $\tilde{\psi}_2(x)$ are MILP-representable, then \eqref{eq:error_robust_MPC_verification} is equivalent to an MILP because $\mathbb X$ is a polyhedron and because $\|\cdot\|_\infty$ is MILP-representable thanks to Lemma~\ref{lem:inf_norm_MILP}.
\end{proof}

\begin{remark}
    In practice, \eqref{eq:error_robust_MPC_verification} is solved with $\hat{\gamma}=0$, and the value of the decision variable~$\tau$ at optimality determines the worst-case error $\gamma$ between~$\psi_1(\cdot)$ and~$\tilde \psi_2(\cdot)$.
\end{remark}
\end{diffadd}

\section{Stability Verification of Approximate MPC} \label{sec:stability_verification}

In Section~\ref{subsec:_verification_robust} we verified the stability of an approximate policy against a robustly stable MPC policy by checking an inequality involving the worst-case approximation error. In contrast, in this section we show that one can also directly verify a Lyapunov decrease condition to show stability.  To this end, consider again the linear dynamical system~\eqref{eq:linear_system}, and assume that the state $x \in \mathbb{R}^n$ and the input $u \in \mathbb{R}^m$ are constrained to lie in the polytopic sets $\mathbb{X} \subseteq \mathbb{R}^n$ and $\mathbb{U} \subseteq \mathbb{R}^m$, respectively. Consider then the finite-horizon MPC problem
\begin{subequations} \label{eq:mpc_scheme}
\begin{align}
    J^{\star}(x)= \min_{\mathbf{x}, \mathbf{u}} ~ & J(\mathbf{x}, \mathbf{u}) \\
    \text {s.t.}\, ~ & x_{i+1}=A x_{i}+B u_{i} ~ \forall i=0, \ldots, N-1  \\
    & x_{i} \in \mathbb{X},\; u_{i} \in \mathbb{U} \quad~~ \forall i=0, \ldots, N-1 \label{eq:mpc_scheme_state_input_constraints} \\
    & x_{N} \in \mathbb{X}_{N} \\
    & x_{0}=x
\end{align}
\end{subequations}
with cost function
\begin{equation*}
    J(\mathbf{x}, \mathbf{u}) = \sum_{i=0}^{N-1} \ell(x_i,u_i) + V_N(x_N).
\end{equation*}

\begin{diffadd}
Throughout the paper we will assume that $\mathbb X$, $\mathbb U$ and $\mathbb X_N$ contain a neighborhood of the origin and that the stage cost function $\ell (x,u)$ as well as the terminal cost function $V_N(x)$ are positive definite. For ease of notation, we henceforth use
\begin{equation*}
    \mathcal{F}(x) = 
    \left\{
    (\mathbf{x}, \mathbf{u}) \;\middle|\; \begin{aligned}
    & x_{i+1}=A x_{i}+B u_{i}~\forall i=0, \ldots, N-1, \\
    & x_{i} \in \mathbb{X},\; u_{i} \in \mathbb{U}\quad ~\;\,\forall i=0, \ldots, N-1, \\
    & x_{N} \in \mathbb{X}_{N},\; x_0 = x
    \end{aligned}
    \right\}
\end{equation*}
as a shorthand for the feasible set of problem \eqref{eq:mpc_scheme}.

\begin{assumption} \label{as:terminal_set}
There exists a linear control law $\psi(x)=Kx$ such that the terminal set $\mathbb{X}_{N} \subseteq \mathbb{X}$ is a polytopic invariant set for the system $x^+=Ax+BKx$ with state and input constraints \eqref{eq:mpc_scheme_state_input_constraints} \cite{borrelli2017}, and $V_N(x)$ is a Lyapunov function for the same system that decreases by one stage cost each time step, that is, $V_N(x^+)-V_N(x) \leq \ell(x,Kx)$ for every $x\in\mathbb{X}_N$.
\end{assumption}

Assumption~\ref{as:terminal_set} is satisfied if the stage cost is quadratic, i.e.,
\begin{equation*}
    \ell(x_i,u_i) = x_i^T Q x_i + u_i^T R u_i
\end{equation*}
with $Q \in \mathbb{S}^n_{++}$ and $R \in \mathbb{S}^m_{++}$, in which case the optimal value function of the LQR problem corresponding to~\eqref{eq:mpc_scheme} serves as the desired Lyapunov function $V_N(x)$ and the optimal LQR controller $\psi(x)=Kx$ serves as the desired linear control law. Then, the control invariant set $\mathbb{X}_N$ can be calculated using the MPT toolbox \cite{herceg2013}, and the decrease condition for $V_N(x)$ is fulfilled. To see this, consider the LQR problem
\begin{equation} \label{eq:LQR_prob}
\begin{aligned}
    V_N(x) = %
    \min_{\mathbf{x}, \mathbf{u}} \quad & \sum_{i=0}^{\infty} \ell(x_i,u_i) \\
    \text { s.t.} \quad & x_{i+1}=A x_{i}+B u_{i}~ \forall i=0, \ldots, N-1 \\
    & x_0 = x
\end{aligned}
\end{equation}
corresponding to~\eqref{eq:mpc_scheme}. 
Clearly, if $\psi(x)=Kx$ is an optimal policy for the infinite-horizon problem~\eqref{eq:LQR_prob}, then $V_N(x^+)=V_N(Ax+BKx)=V_N(x)-\ell(x,Kx)$. In addition, it is well known that $V_N(x)$ is a Lyapunov function for the linear system $x^+=Ax+BKx$; see, e.g., \cite{mayne2000}.

Consider now the optimal MPC policy $\psi_1(x)$ that maps any state~$x$ to an optimal input~$u_0$ of problem~\eqref{eq:mpc_scheme}, and consider an MILP-representable approximate MPC controller $\psi_2(x)$ of $\psi_1(x)$ that is cheap to evaluate (e.g., a NN).
In the following, we construct bilevel programs for the verification of the stability of the approximate MPC controller $\psi_2(x)$. In Section~\ref{subsec:MIP_formulation}, we then show that these bilevel programs are equivalent to MIQPs that are amenable to numerical solution.

\subsection{Sufficient Condition for Lyapunov Decrease} \label{subsec:sufficient}

The next lemma establishes a sufficient condition for the approximate control law $\psi_2(x)$ to be stabilizing for a given initial state set. The condition is inspired by \cite{jones2009} but has been extended to allow for the verification of asymptotic stability.

\begin{lemma} \label{thm:sufficient_stability}
If Assumption~\ref{as:terminal_set} holds, $\psi_2(x)$ is a control law defined over a neighborhood $\mathbb{X}_0\subseteq \mathbb X$ of~$0$ and $\epsilon > 0$, then the optimal value function $J^\star(x)$ of the MPC problem~\eqref{eq:mpc_scheme} is a Lyapunov function for the system $x^+=Ax+B\psi_2(x)$ on~$\mathbb X_0$
provided that for every $x_0 \in \mathbb{X}_0$ there exists~$(\mathbf{x}, \mathbf{u}) \in \mathcal{F}(x_0)$ with $u_0=\psi_2(x_0)$ that satisfies the condition
\begin{equation} \label{eq:sufficient_stability}
    J(\mathbf{x}, \mathbf{u}) - J^\star(x_0) \leq \ell(x_0, u_0) - \epsilon\|x_0\|_2^2. %
\end{equation}
\end{lemma}

\begin{proof}
Choose any $x_0\in \mathbb{X}_0$. By assumption, there exists a sequence $(x_0,\dots,x_N,u_0,\dots,u_{N-1})$ with $u_0=\psi_2(x_0)$ that is feasible in~\eqref{eq:mpc_scheme} for $x=x_0$ and satisfies condition~\eqref{eq:sufficient_stability}. Hence, the shifted sequence $(x_1,\dots,x_{N+1}, u_1,\dots,u_N)$ with $u_N=Kx_N$ and $x_{N+1}=Ax_N+BKx_N$ is feasible in~\eqref{eq:mpc_scheme} for $x=x_1$ thanks to Assumption~\ref{as:terminal_set}, which asserts that $\mathbb{X}_N$ is an invariant set for the system $x^+=Ax+BKx$. Evaluating the cost of this shifted sequence in problem~\eqref{eq:mpc_scheme} then yields
\begin{align*}
    J^{\star}\left(x_{1}\right) & \leq \sum_{i=1}^{N} \ell\left(x_{i}, u_{i}\right) + V_{N}\left(x_{N+1}\right)\\
    &= \sum_{i=0}^{N-1} \ell\left(x_{i}, u_{i}\right) -\ell(x_0,u_0)+V_{N}\left(x_{N}\right)\\ 
    &\quad + V_{N}\left(x_{N+1}\right)-V_{N}\left(x_{N}\right)+\ell\left(x_{N}, u_{N}\right) \\ 
    & \leq J^{\star}(x_0) - \epsilon\|x_0\|_2^2.
\end{align*}
Here, the second inequality holds again because of Assumption~\ref{as:terminal_set}, which requires that $V_N(x^+)-V_N(x) \leq \ell(x,Kx)$ for every $x\in\mathbb{X}_N$, and because of \eqref{eq:sufficient_stability}. As $x_0\in \mathbb{X}_0$ was chosen arbitrarily and as $\varepsilon>0$, the optimal value function $J^\star(x)$ thus constitutes indeed a Lyapunov function for the approximate closed-loop system $x^+=Ax+B\psi_2(x)$ on the set~$\mathbb{X}_0$.
\end{proof}

Condition~\eqref{eq:sufficient_stability} is not easy to verify but is equivalent to the requirement that for every initial state $x_0 \in \mathbb{X}_0$ there exists a sequence $(\mathbf{x}, \mathbf{u}) \in \mathcal{F}(x_0)$ with $u_0=\psi_2(x_0)$ and
\begin{equation*}
    0 \leq \ell(x_0,u_0) - \epsilon \|x_0\|_2^2 - J(\mathbf{x}, \mathbf{u}) + J^\star(x_0).
\end{equation*}
This condition is satisfied if and only if
\begin{equation*}
    \inf_{x_0\in\mathbb{X}_0} \sup_{\substack{(\mathbf{x}, \mathbf{u}) \in \mathcal{F}(x_0) \\ u_0 = \psi_2(x_0)}} \ell(x_0,u_0) - \epsilon \|x_0\|_2^2 -J(\mathbf{x}, \mathbf{u}) + J^\star(x_0)
\end{equation*}
is non-negative. This is only possible if for all $x_0\in\mathbb X_0$ there exists $(\mathbf{x}, \mathbf{u}) \in \mathcal{F}(x_0)$ with $u_0=\psi_2(x_0)$. Otherwise, the minimization player in the above zero-sum game can force the value of the game to~$-\infty$ by selecting an~$x_0$ that makes the inner maximization problem infeasible, in which case condition~\eqref{eq:sufficient_stability} fails to hold. From now on, we may thus assume without loss of generality that for all $x_0\in\mathbb X_0$ there exists $(\mathbf{x}, \mathbf{u}) \in \mathcal{F}(x_0)$ with $u_0=\psi_2(x_0)$. Hence, both the infimum and the supremum in the above min-max problem are attained. By recalling that $J^\star(x_0) = \min_{(\tilde{\mathbf{x}},\tilde {\mathbf{u}}) \in \mathcal{F}(x_0)} J(\mathbf{x}, \mathbf{u})$, the above min-max problem can then be recast as
\begin{equation} \label{eq:min_max}
\begin{aligned}
    &\min_{\substack{x_0\in\mathbb{X}_0 \\ (\tilde{\mathbf{x}}, \tilde{\mathbf{u}}) \in \mathcal{F}(x_0)}} J(\tilde{\mathbf{x}}, \tilde{\mathbf{u}}) \\
    & \qquad + \max_{\substack{(\mathbf{x}, \mathbf{u}) \in \mathcal{F}(x_0) \\ u_0 = \psi_2(x_0)}} \ell(x_0,u_0) - J(\mathbf{x}, \mathbf{u}) - \epsilon \|x_0\|_2^2.
\end{aligned}
\end{equation}
The standard approach in robust optimization to simplify the min-max problem~\eqref{eq:min_max} would be to dualize the inner maximization problem. Unfortunately, this leads to a minimization problem with a bilinear term in the objective function, which can not be easily linearized. The resulting minimization problem is almost impossible to solve with a branch and bound algorithm. Instead, we reformulate \eqref{eq:min_max} as the bilevel program
\begin{equation}
\label{eq:stability_opt_org}
\begin{aligned}
\min &~~ J(\tilde{\mathbf{x}}, \tilde{\mathbf{u}}) + \ell(x_0,u_0) - J(\mathbf{x}, \mathbf{u}) - \epsilon \|x_0\|_2^2 \\
\text{s.t.} \; &~~  x_0\in\mathbb{X}_0, ~ (\tilde{\mathbf{x}}, \tilde{\mathbf{u}}) \in \mathcal{F}(x_0) \\
& ~~ (\mathbf{x}, \mathbf{u}) \in \left\{ \!\!\begin{array}{r@{\;}l}
\arg\min & J(\bar{\mathbf{x}}, \bar{\mathbf{u}}) \\
\text{s.t.}\, & (\bar{\mathbf{x}}, \bar{\mathbf{u}}) \in \mathcal{F}(x_0)\\ & \bar u_0 = \psi_2(x_0).
\end{array}\right.
\end{aligned}
\end{equation}
By construction, condition~\eqref{eq:sufficient_stability} is thus fulfilled if and only if the optimal value of the bilevel program~\eqref{eq:stability_opt_org} is non-negative. However, solving bilevel programs of the form~\eqref{eq:stability_opt_org} is still challenging. In Section~\ref{subsec:MIP_formulation} we thus propose a reformulation of~\eqref{eq:stability_opt_org} as an MIQP. Compared to dualizing the min-max problem \eqref{eq:min_max}, we found that the mixed-integer reformulations described in Section~\ref{subsec:MIP_formulation} are more efficiently solvable in practice. We believe that this is the case because the root relaxation is bounded due to only having bounded primal variables in the objective. Dualizing problem \eqref{eq:min_max} introduces dual variables in the objective, which are generally unbounded.

Even if $J^\star(x)$ is a Lyapunov function on the set $\mathbb{X}_0$, the system may still fail to be stable if it is not invariant on $\mathbb{X}_0$. Next, we thus establish conditions that guarantee stability.
\begin{theorem} \label{thm:set_stability} Let $\mathbb{X}_0\subseteq \mathbb X$ be a neighborhood of~$0$ such that for all $x_0\in\mathbb X_0$ there exists $(\mathbf{x}, \mathbf{u}) \in \mathcal{F}(x_0)$ with $u_0=\psi_2(x_0)$. If Assumption~\ref{as:terminal_set} holds, $\psi_2(x)$ is a continuous control law defined over $\mathbb X_0$ with $\psi_2(0)=0$ and the optimal value of~\eqref{eq:stability_opt_org} is non-negative for some $\epsilon>0$, then there exists a neighborhood $\mathbb{X}_0^{\psi_2}\subseteq \mathbb X_0$ of~$0$ such that the closed-loop system $x^+=Ax+B\psi_2(x)$ converges asymptotically to 0 on $\mathbb{X}_0^{\psi_2}$.
\end{theorem}
\begin{proof}
As the optimal value of~\eqref{eq:stability_opt_org} is non-negative, condition~\eqref{eq:sufficient_stability} is satisfied. Hence, all assumptions of Lemma~\ref{thm:sufficient_stability} hold, implying that~$J^\star(x)$ is a Lyapunov function for the system $x^+=Ax+B\psi_2(x)$ on~$\mathbb{X}_0$. Next, recall that the sets $\mathbb{X}$, $\mathbb{U}$, $\mathbb{X}_N$ and $\mathbb{X}_0$ all contain a neighborhood of the origin, $\psi_2(x)$ is continuous and~$\psi_2(0)=0$. %
In addition, as $\ell(x,u)$ and $V_N(x)$ are positive definite, one can show that $J^\star(x)$ is also positive definite and that it has sublevel sets that are arbitrarily small neighborhoods of~$0$. We may thus define~$\mathbb X_0^{\psi_2}$ as a sublevel set of~$J^\star(x)$ that is both a neighborhood of~$0$ and a subset of~$\mathbb{X}_0$. By construction, $\mathbb X_0^{\psi_2}$ is an invariant set, and the closed-loop system $x^+=Ax+B\psi_2(x)$ converges to 0 on $\mathbb X_0^{\psi_2}$.
\end{proof}

\subsection{Direct Verification of Lyapunov Decrease} \label{subsec:direct}

Instead of verifying the sufficient condition \eqref{eq:sufficient_stability}, we can also directly verify the Lyapunov decrease condition
\begin{equation} \label{eq:direct_stability}
    J^\star(x^+)-J^\star(x) \leq - \epsilon\|x\|_2^2 \quad \forall x\in\mathbb{X}_0
\end{equation}
with $x^+=Ax+B\psi_2(x)$. Note that~\eqref{eq:direct_stability} makes only sense if both $J^\star(x)$ and $J^\star(x^+)$ are finite for all $x\in\mathbb X_0$. From now on, we may thus assume without much loss of generality that for all $x_0\in\mathbb{X}_0$ there exist $(\mathbf{x}, \mathbf{u}) \in \mathcal{F}(x_0^+)$ and $(\tilde{\mathbf{x}}, \tilde{\mathbf{u}}) \in \mathcal{F}(x_0)$. In analogy to Section~\ref{subsec:sufficient}, we can then reformulate the direct condition \eqref{eq:direct_stability} in terms of the bilevel program
\begin{equation}
\label{eq:stability_direct_opt_org}
\begin{aligned}
\min  &~~ J(\tilde{\mathbf{x}}, \tilde{\mathbf{u}})- J(\mathbf{x}, \mathbf{u}) - \epsilon \|x_0\|^2_2 \\
\text{s.t.} \; &~~  x_0\in\mathbb{X}_0,~  (\tilde{\mathbf{x}}, \tilde{\mathbf{u}}) \in \mathcal{F}(x_0)\\
& ~~ (\mathbf{x}, \mathbf{u}) \in \left\{ \!\! \begin{array}{r@{\;}l}
\arg \min & J(\bar{\mathbf{x}}, \bar{\mathbf{u}})\\
\text{s.t.}\, & (\bar{\mathbf{x}}, \bar{\mathbf{u}}) \in \mathcal F(\bar x_0) \\ & \bar x_0=Ax_0+B\psi_2(x_0).
\end{array}\right.
\end{aligned}
\end{equation}
Indeed, condition~\eqref{eq:direct_stability} holds if and only if the optimal value of the bilevel program~\eqref{eq:stability_direct_opt_org} is non-negative. Note that, by constructions, the lower level problem in~\eqref{eq:stability_direct_opt_org} is feasible for every~$x_0\in\mathbb X_0$. Note also that the bilevel programs~\eqref{eq:stability_opt_org} and~\eqref{eq:stability_direct_opt_org} are very similar, that is, they only differ with regard to the objective function of the upper level problem and the feasible set of the lower level problem.
\begin{corollary} \label{cor:direct}
Let $\mathbb{X}_0\subseteq \mathbb X$ be a neighborhood of~$0$ such that for all $x_0\in\mathbb{X}_0$ there exist $(\mathbf{x}, \mathbf{u}) \in \mathcal{F}(x_0^+)$ and $(\tilde{\mathbf{x}}, \tilde{\mathbf{u}}) \in \mathcal{F}(x_0)$. If Assumption~\ref{as:terminal_set} holds, $\psi_2(x)$ is a continuous control law on $\mathbb{X}_0$ with $\psi_2(0)=0$ and the optimal value of~\eqref{eq:stability_direct_opt_org} is non-negative for some $\epsilon>0$, then there exists a neighborhood $\mathbb X_0^{\psi_2} \subseteq \mathbb{X}_0$ of~$0$ such that the closed-loop system $x^+=Ax+B\psi_2(x)$ converges asymptotically to 0 on $\mathbb X_0^{\psi_2}$.
\end{corollary}

The proof of Corollary~\ref{cor:direct} widely parallels that of Theorem~\ref{thm:set_stability} and is therefore omitted for the sake of brevity.

\subsection{Approximating the Stable Region} \label{subsec:outer_approx}

From Theorem~\ref{thm:set_stability} we know that the approximate controller $\psi_2(x)$ is stable on a neighborhood~$\mathbb X_0^{\psi_2}$ of~0 provided that the optimal value of the bilevel program~\eqref{eq:stability_opt_org} is non-negative. Unfortunately, there is no explicit analytical or efficient algorithmic characterization of~$\mathbb X_0^{\psi_2}$ in general. We thus resort to computing outer and inner approximations. To construct a convex polyhedral outer approximation for~$\mathbb X_0^{\psi_2}$, recall first that~$\mathbb X_0^{\psi_2}\subseteq \mathbb X_0$ and that for all $x_0\in\mathbb X_0$ there exists $(\mathbf{x}, \mathbf{u}) \in \mathcal{F}(x_0)$ with $u_0=\psi_2(x_0)$. We can thus enclose~$\mathbb X_0^{\psi_2}$ in 
\[
    \Omega=\big\{x_0\in\mathbb X \; \big| \; \exists (\mathbf{x}, \mathbf{u}) \in \mathcal{F}(x_0),\; u_0=\psi_2(x_0)\big\},
\]
which in turn can be enclosed in the convex polyhedron
\begin{equation*}
    \overline{\Omega} = \left\{ x \in \mathbb{R}^{n} \;\middle|\; Cx \leq c \right\},
\end{equation*}
where~$C \in \mathbb{R}^{n_C \times n}$ is a fixed matrix, whose rows represent the $n_C$ normal vectors of the polyhedron's facets. The vector $c \in \mathbb{R}^{n_C}$ is constructed by solving $n_C$ optimization problems
\begin{equation}
\label{eq:outer_approx_MILPs}
c_i= \max_{x_0 \in \Omega} C_i x_0
\end{equation}
for $i=1,\dots, n_C$. This recipe yields the smallest polyhedron~$\overline\Omega$ corresponding to the shape matrix~$C$ that contains $\Omega$. For example, if we select $C=[I_n, -I_n]^T$, then $\overline \Omega$ coincides with the smallest rectangular box containing~$\Omega$.

To construct an ellipsoidal subset of the stability region, we first consturct a large ball in $\Omega$ around the origin of $\mathbb R^n$. This is difficult because of the nonlinear constraint $u_0=\psi_2(x_0)$ in the definition of $\Omega$. However, if $\psi_2(0)=0$ and $\psi_2(x)$ is Lipschitz continuous with Lipschitz constant $L$, then we have
\[
    \|u_0\|_2 = \|\psi_2(x_0)-\psi_2(0)\|_2 \leq L \|x_0\|_2.
\]
Note that MILP techniques akin to those developed in this paper can be used to compute the Lipschitz constant $L$ if $\psi_2(x)$ is a ReLU NN \cite{jordan2020}. Next, consider the feasible set of the MPC problem~\eqref{eq:mpc_scheme} projected to the first state and input, that is,
\[
    \mathcal F_{x,u}=\big\{(x,u)\in\mathbb R^n\times\mathbb R^m\;\big|\;\exists (\mathbf{x}, \mathbf{u}) \in \mathcal{F}(x),\; u_0=u\big\},
\]
and note that $\mathcal F_{x,u}$ is the projection of a convex polyhedron, thus constituting a lower-dimensional convex polyhedron. By using Fourier-Motzkin elimination, one can find a representation for $\mathcal F_{x,u}$ in terms of linear inequalities, that is, we can construct $D\in\mathbb R^{n_{x,u}\times n}$, $E\in\mathbb R^{n_{x,u}\times m}$ and $d\in\mathbb R^{n_{x,u}}$ with
\[
    \mathcal F_{x,u}=\big\{(x,u)\in\mathbb R^n\times\mathbb R^m\;\big|\; Dx+Eu\leq d\big\}.
\]
In fact, this explicit representation of $\mathcal F_{x,u}$ can be efficiently computed with the MPT toolbox \cite{herceg2013}. We may thus conclude that $\Omega$ is guaranteed to contain the ball $\underline\Omega$ around~$0$ of radius
\begin{equation*}
\begin{aligned}
   &\max_{r \geq 0} \big\{ r\;\big|\; (x_0,u_0)\in \mathcal F_{x,u}~\forall \|x_0\|_2\leq r,\; \|u_0\|_2 \leq L r\big\} \\
   & =\left\{ \begin{array}{c@{\;}l}
        \displaystyle \max_{r \geq 0} & r  \\
        \text{s.t.} & Dx_0+Eu_0\leq d ~ \forall \|x_0\|_2\leq r, \|u_0\|_2 \leq L r
   \end{array}\right.\\
   & =\left\{ \begin{array}{c@{\;}l}
        \displaystyle \max_{r \geq 0} & r  \\
        \text{s.t.} & (\|D_i\|_2 +\|E_i\|_2L)r \leq d_i ~ \forall i=1,\ldots,n_{x,u}
   \end{array}\right.\\
   &= \min_{1,\ldots,n_{x,u}} d_i/(\|D_i\|_2 +\|E_i\|_2L)=r^\star.
\end{aligned}
\end{equation*}
From Theorem~\ref{thm:set_stability} we know that $J^\star(x)$ is a Lyapunov function on $\Omega$ and, consequently, also on~$\underline \Omega\subseteq \Omega$. While $\underline\Omega$ may not be invariant under $\psi_2(x)$, the defining properties of Lyapunov functions imply that any sublevel set of $J^\star(x)$ contained in~$\underbar{$\Omega$}$ must be invariant. As $J^\star(x)$ is convex piecewise quadratic and, in particular, coincides with~$V_N(x)$ on $\mathbb X_N$, its smallest sublevel sets are concentric ellipsoids around~0 \cite[Chapter~7.3]{rawlings2017}. One can thus define $\mathbb X_0^{\psi_2}$ as the largest of these sublevel sets contained in $\underbar{$\Omega$}\cap \mathbb X_N$, which can easily be computed. By construction, the system $Ax+B\psi_2(x)$ is asymptotically stable on~$\mathbb X_0^{\psi_2}$ (but~$\mathbb X_0^{\psi_2}$ constructed in this manner is generally a strict subset of the entire stability region).
\end{diffadd}

\subsection{MIP Reformulations of Bilevel Programs} \label{subsec:MIP_formulation}

The bilevel program~\eqref{eq:stability_opt_org} derived in Section~\ref{subsec:sufficient} cannot be solved directly by standard solvers. We will now show, however, that it can sometimes be reformulated as an MIQP. To this end, we assume from now on that the stage cost function $\ell(x,u)$ as well as the terminal cost function $V_N(x)$ are convex and quadratic. In addition, we assume that the sets $\mathbb X$, $\mathbb U$, $\mathbb X_N$ and $\mathbb X_0$ are convex polyhedra. In this case, the lower level problem in~\eqref{eq:stability_opt_org} aconstitutes a convex QP parameterized by $x_0$ and $\psi_2(x_0)$ and can be represented abstractly as
\begin{equation} \label{eq:qp_reform}
\begin{aligned}
\min_{z} \quad & \frac{1}{2} z^{T} P z + q^Tz \\
\text {s.t.}\, \quad & Az=B[x_0^T, \psi_2(x_0)^T]^T +b \\
& Fz \leq g,
\end{aligned}
\end{equation}
with $z\in \mathbb{R}^{n_z}$ representing the decision variables $(\bar{\mathbf{x}}, \bar{\mathbf{u}})$ and with constant matrices $P \in \mathbb{S}^{n_z}_{++}$, $q \in \mathbb{R}^{n_z}$, $A \in \mathbb{R}^{n_\text{eq} \times n_z}$, $B \in \mathbb{R}^{n_\text{eq} \times (n+m)}$, $b \in \mathbb{R}^{n_\text{eq}}$, $F \in \mathbb{R}^{n_\text{ineq} \times n_z}$, and $g \in \mathbb{R}^{n_\text{ineq}}$.
By Lemma~\ref{lem:qp_MILP}, the unique minimizer of~\eqref{eq:qp_reform} is therefore MILP-representable, and the bilevel program~\eqref{eq:stability_opt_org} is equivalent to
\begin{equation} \label{eq:stability_opt_mip}
\begin{aligned}
\min \quad & J(\tilde{\mathbf{x}}, \tilde{\mathbf{u}})+\ell(x_0,u_0) - \frac{1}{2} z^{T} P z - \tilde{q}^T\tilde{z}-\epsilon\|x_0\|_2^2\\
\text {s.t.}\,\quad & x_0 \in \mathbb{X}_0,~(\tilde{\mathbf{x}}, \tilde{\mathbf{u}}) \in \mathcal{F}(x_0),~z\in\mathbb R^{n_z},~ u_0\in\mathbb R^m\\
& \mu\in\mathbb R^{n_{\text{eq}}},~ \lambda\in\mathbb R^{n_{\text{ineq}}},~ \beta \in\{0,1\}^{n_{\text{ineq}}} \\
& (x_0,u_0) \in \gr({\psi_2})\\
&0 = Az - B[x_0^T, u_0^T]^T - b \\
&0 = Pz + q + A^T\mu + F^T\lambda \\
&0 \leq \lambda_i \leq M\beta_i\quad \forall i=1,\dots,n_\text{ineq} \\
&0 \leq g_i - F_iz \leq M(1-\beta_i)\quad\forall i=1,\dots,n_\text{ineq},
\end{aligned}
\end{equation}
where $\mu$ and $\lambda$ represent the Lagrange multipliers associated with the equality and inequality constraints of the lower level problem~\eqref{eq:qp_reform}, respectively, and $\beta$ is the binary vector used to linearize the complementary slackness conditions. \begin{diffadd}The auxiliary variable~$u_0$ serves as a placeholder for~$\psi_2(x_0)$, and the abstract constraint $(x_0,u_0) \in \gr({\psi_2})$ captures the requirement~$u_0=\psi_2(x_0)$. If~$\psi_2(x_0)$ is MILP-representable (e.g., if~$\psi_2(x_0)$ is a ReLU NN), then \eqref{eq:stability_opt_mip} is equivalent to an indefinite MIQP, which can be readily solved with commercial solvers such as Gurobi~\cite{gurobi}. Similarly, the bilevel program~\eqref{eq:stability_direct_opt_org} can also be reformulated as an MIQP, and the optimization problems~\eqref{eq:outer_approx_MILPs} that need to be solved to construct~$\overline \Omega$ can be reformulated as MILPs. Details are omitted for brevity.\end{diffadd}

\begin{remark}
As its objective function is indefinite quadratic, problem \eqref{eq:stability_opt_mip} is considerably harder to solve than an MILP or an MIQP with a convex QP relaxation. Modern solvers such as Gurobi 9.0 can solve problems of the form~\eqref{eq:stability_opt_mip}, however, by translating the indefinite cost into bilinear constraints and then using cutting plane and spatial branching techniques.
\end{remark}

\section{Case Study: DC-DC Power Converter} \label{sec:case_study}

In this section, we consider the DC-DC power converter from \cite{maddalena2021} as a case study. The original linear MPC controller is approximated via a piecewise affine NN. Evaluating the NN is considerably cheaper, allowing the approximate control policy to be deployed on a low cost microcontroller operating at 80MHz and with a sampling frequency of 10 kHz. Simulation results show the effectiveness of the approximate control in some cases, however, the work gives no guarantee of stability.
Using the techniques introduced above, we can show that the NN controller is indeed stable for the closed-loop system, and give an outer approximation of the stability region.

Additionally, we design a tube MPC robust to input disturbances. The robust controller is then approximated using a NN. Using the formulation in Section~\ref{subsec:_verification_robust} we can guarantee the closed-loop system controlled by the NN satisfies constraints and converges to the minimum robust invariant set.

The model of the DC-DC converter is linearized and discretized, giving us the following two-state $x = (i_L, v_O)$ (current and voltage), and one-input (duty cycle) linear system
\begin{equation*}
    x^+ = Ax + Bu = \begin{bmatrix} 0.971 & -0.010 \\ 1.732 & 0.970 \end{bmatrix}x + \begin{bmatrix} 0.149 \\ 0.181 \end{bmatrix}u.
\end{equation*}
A more detailed derivation of the model dynamics along with additional details on the physical system can be found in \cite{maddalena2021}.

\subsection{Nominal MPC} \label{subsec:nominal_mpc}

The problem of control design for tracking current and voltage references in the DC-DC converter can be formulated as a linear MPC controller. We formulate the MPC controller which regulates to the steady-state $x_\text{eq}=\begin{bmatrix} 0.05 & 5 \end{bmatrix}^T$, and $u_\text{eq}=0.3379$ as
\begin{equation} \label{original_MPC}
\begin{aligned}
    \min_{\mathbf{x}, \mathbf{u}} \quad & \sum_{i=0}^{N-1}\left(\left\|x_{i}-x_{\mathrm{eq}}\right\|_{Q}^{2}+\left\|u_{i}-u_{\mathrm{eq}}\right\|_{R}^{2}\right)\\
    & \quad\quad +\left\|x_{N}-x_{\mathrm{eq}}\right\|_{P}^{2} \\
    \text { s.t.} \quad & \forall i=0, \ldots, N-1, \\
    & x_{i+1}=A x_{i}+B u_{i}, \\
    & \begin{bmatrix} 0 \\ 0 \end{bmatrix} \leq x_i \leq \begin{bmatrix} 0.2 \\ 7 \end{bmatrix}, \quad 0 \leq u_i \leq 1 \\
    & x_{N} \in \mathbb{X}_{N}, \quad x_{0}=x(0),
\end{aligned}
\end{equation}
with horizon length $N=10$, $Q=\diag(90, 1)$, $R=1$, $P$ the solution of the associated discrete-time algebraic Riccati equation, and $\mathbb{X}_N$ the system's maximal invariant set under the LQR policy. The resulting control policy can be seen in Figure~\ref{fig:nominal_mpc}.

\begin{figure}
    \centering
    \resizebox{0.25\textwidth}{!}{\tikzset{every picture/.style={line width=0.75pt}} %

\begin{tikzpicture}[x=0.75pt,y=0.75pt,yscale=-1,xscale=1]

\draw    (100,10) -- (100,27) ;
\draw [shift={(100,30)}, rotate = 270] [fill={rgb, 255:red, 0; green, 0; blue, 0 }  ][line width=0.08]  [draw opacity=0] (10.72,-5.15) -- (0,0) -- (10.72,5.15) -- (7.12,0) -- cycle    ;
\draw  [color={rgb, 255:red, 224; green, 213; blue, 82 }  ,draw opacity=1 ][fill={rgb, 255:red, 255; green, 246; blue, 142 }  ,fill opacity=1 ] (10,30) -- (190,30) -- (190,60) -- (10,60) -- cycle ;
\draw    (100,60) -- (100,77) ;
\draw [shift={(100,80)}, rotate = 270] [fill={rgb, 255:red, 0; green, 0; blue, 0 }  ][line width=0.08]  [draw opacity=0] (10.72,-5.15) -- (0,0) -- (10.72,5.15) -- (7.12,0) -- cycle    ;
\draw  [color={rgb, 255:red, 224; green, 213; blue, 82 }  ,draw opacity=1 ][fill={rgb, 255:red, 255; green, 246; blue, 142 }  ,fill opacity=1 ] (10,80) -- (190,80) -- (190,120) -- (10,120) -- cycle ;
\draw    (100,120) -- (100,137) ;
\draw [shift={(100,140)}, rotate = 270] [fill={rgb, 255:red, 0; green, 0; blue, 0 }  ][line width=0.08]  [draw opacity=0] (10.72,-5.15) -- (0,0) -- (10.72,5.15) -- (7.12,0) -- cycle    ;
\draw  [color={rgb, 255:red, 224; green, 213; blue, 82 }  ,draw opacity=1 ][fill={rgb, 255:red, 255; green, 246; blue, 142 }  ,fill opacity=1 ] (10,140) -- (190,140) -- (190,170) -- (10,170) -- cycle ;
\draw    (100,170) -- (100,187) ;
\draw [shift={(100,190)}, rotate = 270] [fill={rgb, 255:red, 0; green, 0; blue, 0 }  ][line width=0.08]  [draw opacity=0] (10.72,-5.15) -- (0,0) -- (10.72,5.15) -- (7.12,0) -- cycle    ;
\draw  [color={rgb, 255:red, 224; green, 213; blue, 82 }  ,draw opacity=1 ][fill={rgb, 255:red, 255; green, 246; blue, 142 }  ,fill opacity=1 ] (10,190) -- (190,190) -- (190,220) -- (10,220) -- cycle ;
\draw    (103,220) -- (103,237) ;
\draw [shift={(103,240)}, rotate = 270] [fill={rgb, 255:red, 0; green, 0; blue, 0 }  ][line width=0.08]  [draw opacity=0] (10.72,-5.15) -- (0,0) -- (10.72,5.15) -- (7.12,0) -- cycle    ;

\draw (85.5,18.5) node  [font=\footnotesize]  {$x$};
\draw (100,45) node  [font=\footnotesize]  {$Fx+f$};
\draw (85.5,68.5) node  [font=\footnotesize]  {$y_{1}$};
\draw (101.5,99.5) node  [font=\footnotesize]  {$ \begin{array}{l}
\arg\min \|Hz+y_{1} \|_2^{2} +\varepsilon \|z\|_2^{2}\\
\ \ \ \ \ \text{ s.t. } z\geq 0
\end{array}$};
\draw (85.5,128.5) node  [font=\footnotesize]  {$y_{2}$};
\draw (100,155) node  [font=\footnotesize]  {$Gy_{2} +g$};
\draw (85.5,178.5) node  [font=\footnotesize]  {$y_{3}$};
\draw (100,205) node  [font=\footnotesize]  {$\text{sat}( y_{3})$};
\draw (88.5,228.5) node  [font=\footnotesize]  {$u$};

\end{tikzpicture}}
    \caption{The piecewise affine NN architecture.}
    \label{fig:PWA-NN}
\end{figure}
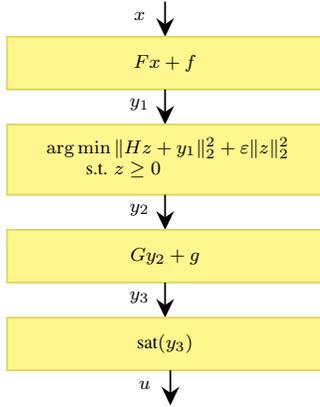

We chose the same architecture for the approximate controller as in \cite{maddalena2021}, which can be seen in Figure~\ref{fig:PWA-NN}, with the appropriate dimensions of the layers fully determined by the size of the optimization variable $z \in \mathbb{R}^{n_z}$, which was chosen as $n_z=3$. Note that this is not a classic NN architecture, since the second layer is the solution of a parametric QP. This shows the flexibility of our toolbox, since we also allow more unconventional layer structures.

\begin{diffadd}
The idea of the second layer is to approximate the original MPC policy, which is a parametric QP, with a smaller parametric QP where we can control the complexity with the dimension $n_z$. By leveraging the implicit function theorem, one can optimize the QP parameters. For deployment, the solution map of the QP can then be precalculated like in explicit MPC. Compared to explicit MPC, we can trade off the solution complexity with approximation accuracy. The reader is referred to \cite{maddalena2021} for a more detailed description and analysis of the architecture.
\end{diffadd}

For the training data, 5000 samples of the original controller \eqref{original_MPC} uniformly distributed from the feasible region are taken, and the approximate controller is then trained using Adam \cite{kingma2015} as the optimizer. To formulate the convex optimization layer, cvxpylayers \cite{agrawal2019} was used. The trained simplified controller and the absolute approximation error can be seen in Figure~\ref{fig:nominal_mpc}.

Using the formulation in Section~\ref{sec:error_computation} we find an absolute worst-case approximation error of $\gamma = 0.24$ at $i_L=0$ and $v_O=6.79$. Similarly, using both formulations in Sections \ref{subsec:sufficient} and \ref{subsec:direct}, by solving an indefinite MIQP respectively, we can verify the stability of the closed-loop system, which is the case for this example. With this, we have successfully verified the approximate control policy against the nominal MPC. Additionally, we solve \eqref{eq:outer_approx_MILPs} to get an outer and inner approximation of the stability region $\Omega$, which can be seen in Figure~\ref{fig:nominal_mpc}. Note that the outer approximation is smaller than the feasible region of the MPC controller. This is expected because due to the approximation error, we can not expect to achieve the same stability region. \begin{diffadd}On the other hand, the inner approximation is very conservative, as one would expect, since we are using conservative Lipschitz bounds for its calculation.\end{diffadd}

\begin{figure}
    \centering
    \begin{subfigure}[b]{0.23\textwidth}
        \centering
        \includegraphics[width=\textwidth]{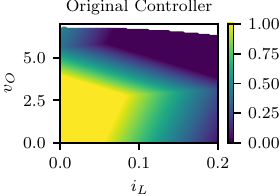}
    \end{subfigure}
    \begin{subfigure}[b]{0.23\textwidth}
        \centering
        \includegraphics[width=\textwidth]{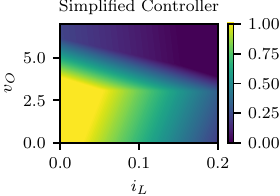}
    \end{subfigure}
    \begin{subfigure}[b]{0.3\textwidth}
        \centering
        \includegraphics[width=\textwidth]{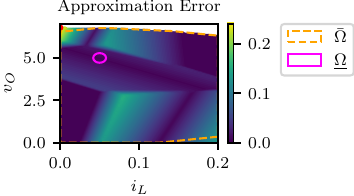}
    \end{subfigure}
    \caption{Original, approximate/simplified control policy (top), and approximation error to the original MPC control policy with outer and inner approximation of stability region (bottom).}
    \label{fig:nominal_mpc}
\end{figure}

\subsection{Robust Tube MPC} \label{subsec:example_robust_mpc}

Starting from the MPC formulation in \eqref{original_MPC}, we formulate a robust tube MPC \cite{mayne2005}. We assume that we have uncertain dynamics
\begin{equation*}
    x^+ = Ax + Bu + Bw,
\end{equation*}
with disturbance set $\mathbb{W} = \left\{ w \in \mathbb{R} \;\middle|\; -0.1 \leq w \leq 0.1\right\}$. The objective is now to design a robust controller to withstand an input disturbance with amplitude 0.1. The input disturbance is typically associated to the distortion introduced by the switching nature of the converter, but here we use it to verify the stability of the approximate policy.

We extend the nominal MPC formulation \eqref{original_MPC} to the following robust tube MPC formulation
\begin{equation} \label{robust_MPC}
\begin{aligned}
    \min_{\mathbf{z}, \mathbf{v}} \quad & \sum_{i=0}^{N-1}\left(\left\|z_{i}-x_{\mathrm{eq}}\right\|_{Q}^{2}+\left\|v_{i}-u_{\mathrm{eq}}\right\|_{R}^{2}\right)+\left\|z_{N}-x_{\mathrm{eq}}\right\|_{P}^{2} \\
    \text { s.t.} \quad & \forall i=0, \ldots, N-1, \\
    & z_{i+1}=A z_{i}+B v_{i}, \\
    & z_i \in \mathbb{X} \ominus \mathcal{E}, \quad v_i \in \mathbb{U} \ominus K \mathcal{E} \\
    & z_{N} \in \mathbb{X}_{N}, \quad x(0) \in z_{0} \oplus \mathcal{E},
\end{aligned}
\end{equation}
where $\mathcal{E}$ is the minimum robust invariant set with respect to a linear feedback gain $K$, and the control law is given by $\psi_1(x)=K(x-z_0^\star(x))+v_0^\star(x)$. Here $K$ is the feedback gain and $P$ is the solution of the Riccati equation associated with the discrete LQR problem. The policy is robustly stable for input disturbances in $\mathbb{W}$ \cite[Theorem~1]{mayne2005}. The horizon length was increased to $N=20$ to have a bigger feasible region. Note that \eqref{robust_MPC} is a pQP since the sets $\mathbb{X} \ominus \mathcal{E}$, $\mathbb{U} \ominus K \mathcal{E}$, and $\mathcal{E}$ are polyhedra that can be precalculated.

We then approximate the tube MPC with a NN with 2 hidden layers, 50 neurons each, and a saturation layer at the end to clip the input between -1 and 1. Similar as in Section~\ref{subsec:nominal_mpc}, 5000 samples of the tube MPC uniformly in the feasible region are taken and the NN is trained using standard techniques. The resulting NN controller together with the original tube MPC can be seen in Figure~\ref{fig:robust_mpc}.

\begin{figure}
    \centering
    \begin{subfigure}[b]{0.23\textwidth}
        \centering
        \includegraphics[width=\textwidth]{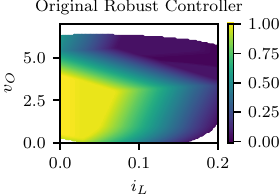}
    \end{subfigure}
    \begin{subfigure}[b]{0.23\textwidth}
        \centering
        \includegraphics[width=\textwidth]{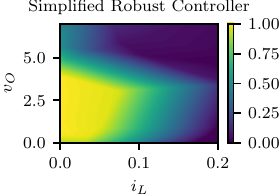}
    \end{subfigure}
    \begin{subfigure}[b]{0.23\textwidth}
        \centering
        \includegraphics[width=\textwidth]{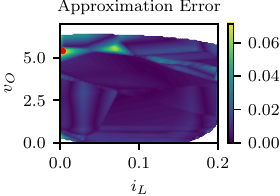}
    \end{subfigure}
    \caption{Original, approximate/simplified robust control policy (top), and approximation error to the original robust MPC control policy (bottom).}
    \label{fig:robust_mpc}
\end{figure}

We apply the formulation from Section~\ref{sec:error_computation} to find the worst case approximation error, and get a value of $\gamma=0.073$. Since we designed our controller to be robust for input perturbations with a maximum magnitude of $0.1$, we have shown that the NN controller satisfies constraints and converges to the minimum robust invariant set in the feasible region of the tube MPC.

The set on which the closed-loop system is stable is directly given by the feasible set of the tube MPC, compared to the nominal MPC where we could only calculate an outer approximation of the stability region. Additionally, to verify stability, we only have to solve an MILP and not an indefinite MIQP, which might be considerably harder to solve. But for larger problem sizes, the solve times are not necessarily larger as we discuss in Section~\ref{subsec:num_exp}.

\subsection{Experimental Validation}

\begin{figure}
    \centering
    \includegraphics[width=0.48\textwidth]{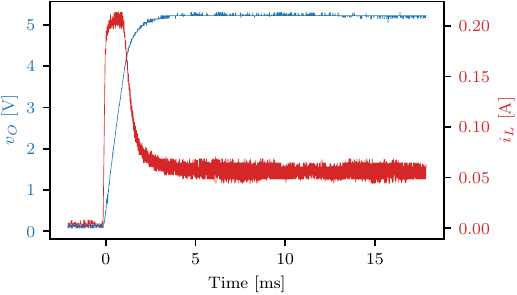}
    \caption{Closed-loop startup response of the approximate MPC.}
    \label{fig:experiment}
\end{figure}

The approximate nominal controller in Section~\ref{subsec:nominal_mpc} has been implemented on an inexpensive STM32L476 microcontroller for a prototype of the described DC-DC power converter. The controller is running at a frequency of 10 kHz, with the execution time of the control law being between 22.0 $\mu$s and 27.5 $\mu$s. The closed-loop startup response can be seen in Figure~\ref{fig:experiment}. We can see that the constraints for both the voltage $v_O$ and the current $i_L$ are satisfied, and an excellent transient with a settling time of 2.33 ms. For more elaborate implementation details, the reader is referred to \cite{maddalena2021}.

\section{Discussion}

\subsection{Comparison to Explicit MPC}

\begin{diffadd}
Similar to explicit MPC, we obtain an explicit solution map of the baseline MPC policy. The main difference lies in the representation and exactness of this mapping. Whereas the explicit MPC solution map is exact, compared to the approximate, but verified, solution map obtained by our method, it comes at the cost of storing polytopic regions of the piecewise linear solution map. Even worse, the number of stored regions grows, in the worst case, exponentially in the dimension of the state and the number of constraints of the baseline MPC policy \cite{alessio2009}. Additionally, the online search for the polytope containing the current state may require excessive processing power or storage space. Although these computational challenges can be mitigated by using search trees \cite{jones2006} or hash tables \cite{bayat2011}, the required memory may remain too large \cite{kvasnica2012}, especially for embedded systems.

The NN approximating the baseline MPC policy in our method does not suffer from this exponential growth. Thus, deploying the policy online does not require excessive amounts of memory and no online search. The inference of the NN is deterministic and does not change based on the current state. This makes it especially suitable for time-critical systems where low deterministic evaluation times are critical. In \cite{maddalena2020}, the authors could reduce the evaluation time and required memory to store the solution map from 12.9~ms and 518~kB for the explicit MPC to 1.5~ms and 17~kB for the NN approximation. Considering this was a relatively simple system with only $n=2$ states and $m=1$ input, it is to be expected that for larger systems the savings in online evaluation time and required memory storage only increases.

In the worst case, the verification method described in Section~\ref{sec:stability_verification} scales worse than the offline calculation of explicit MPC, since it not only exponentially grows with the number of constraints, but also with the number of neurons of the NN. The only advantage is that our method is that we can rely on efficient branch-and-cut methods to reduce the search space, while explicit MPC has to enumerate every polytopic region. But it is worth noting that all this computational complexity is offline, and our method does not suffer from the exponential scaling in the online setting.
\end{diffadd}

\subsection{Comparison to Lipschitz Based Methods}

We compare our method to the approach introduced in \cite{fabiani2021}. The authors provide the following two main results.

\begin{lemma}[{{\cite[Lemma 3.2]{fabiani2021}}}] \label{lem:fabiani2021}
There exist $\zeta_c > 0$ such that, if $\gamma < \zeta_c$, the LTI system in \eqref{eq:linear_system} with NN controller $\psi_2(x)$ converges in finite time to a neighborhood of the origin, for all $x_0 \in \Omega_c$, with $c \coloneqq \max_a \left\{a \geq 0 \mid \Omega_{a} \subseteq \mathbb{X}\right\}$.
\end{lemma}

\begin{theorem}[{{\cite[Theorem 3.4]{fabiani2021}}}] \label{thm:fabiani2021}
Let Assumption~\ref{as:terminal_set} hold. There exist $\zeta_c$ and $\vartheta$ such that, if $\gamma < \zeta_c$ and $\mathcal{L}_{\infty}\left(\psi_1(x)-\psi_2(x), \mathbb{X}_{N}\right) < \vartheta$, and $b \geq 0$ can be chosen so that
$\Omega_b \subseteq \mathbb{X}_N$, then the LTI system in \eqref{eq:linear_system} with NN controller $\psi_2(x)$ converges exponentially to the origin, for all $x_0 \in \Omega_c$, with $c \coloneqq \max_a \left\{a \geq 0 \mid \Omega_{a} \subseteq \mathbb{X}\right\}$.
\end{theorem}

Here $\Omega_a$ denotes the $a$-sublevel set of $J^\star(x)$, and $\mathcal{L}_{\alpha}(F, \mathcal{S})$ is the local $\alpha$-Lipschitz constant over some set $\mathcal{S} \subseteq \mathbb{R}^n$ for a given mapping $F: \mathbb{R}^{n} \rightarrow \mathbb{R}^{m}$.

Lemma~\ref{lem:fabiani2021} states that the NN policy is input-to-state-stable (ISS) if the worst case approximation error is smaller than a constant $\zeta_c$. This does not guarantee asymptotic stability yet, but in Theorem~\ref{thm:fabiani2021} exponential convergence to the origin is shown if the Lipschitz constant of the approximation error is smaller than a constant $\vartheta$.

We are not going in more detail on how to calculate these constants. The reader is referred to \cite{fabiani2021} for more details, but it turns out that these constants become small in practice. Since the NN almost surely does not approximate the baseline control policy arbitrary well in practice, the Lipschitz constant of the approximation error can grow significantly, especially if the NN is over parameterized. To demonstrate this, we consider a simple example of a double integrator. We define the baseline MPC policy
\begin{equation} \label{eq:double_integrator}
\begin{aligned}
    \min_{\mathbf{x}, \mathbf{u}} \quad & \sum_{i=0}^{N-1}\left(\left\|x_{i}\right\|_{Q}^{2}+\left\|u_{i}\right\|_{R}^{2}\right) +\left\|x_{N}\right\|_{P}^{2} \\
    \text { s.t.} \quad & \forall i=0, \ldots, N-1, \\
    & x_{i+1}= \begin{bmatrix} 1 & 1 \\ 0 & 1 \end{bmatrix} x_{i} + \begin{bmatrix} 0 \\ 1 \end{bmatrix} u_{i}, \\
    & \begin{bmatrix} -10 \\ -10 \end{bmatrix} \leq x_i \leq \begin{bmatrix} 10 \\ 10 \end{bmatrix}, \quad -1 \leq u_i \leq 1 \\
    & x_{N} \in \mathbb{X}_{N}, \quad x_{0}=x(0),
\end{aligned}
\end{equation}
with horizon length $N=10$, $Q=\diag(1, 1)$, $R=0.1$, $P$ the solution of the associated discrete-time algebraic Riccati equation, and $\mathbb{X}_N$ the system's maximal invariant set under the LQR policy.

We estimate the largest sublevel set $\Omega_c \subseteq \mathbb{X}$ with $c=180$, giving us $\zeta_{180}=0.1673$. This will only guarantee ISS. To prove exponential stability, we estimate the largest sublevel $\Omega_b \subseteq \mathbb{X}_N$ with $b=1.1$, giving us $\zeta_{1.1}=0.0010$ and $\vartheta=0.0436$.

\begin{figure}
    \centering
    \includegraphics[width=0.48\textwidth]{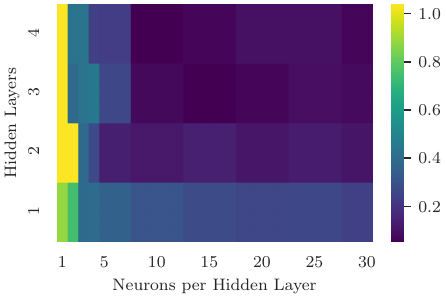}
    \caption{Worst case approximation error $\gamma$.}
    \label{fig:approximation_error_1000}
\end{figure}

\begin{figure}
    \centering
    \includegraphics[width=0.48\textwidth]{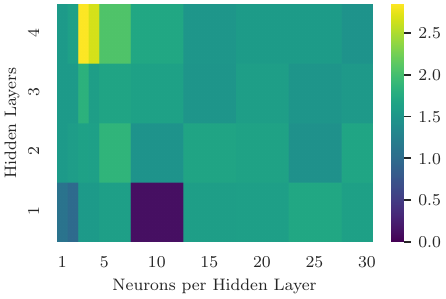}
    \caption{Lipschitz constant $\mathcal{L}_{\infty}\left(\psi_1(x)-\psi_2(x), \mathbb{X}_{N}\right)$ of the approximation error.}
    \label{fig:approximation_error_lipschitz_constant_1000}
\end{figure}

For the training data, 1000 samples of the baseline controller \eqref{eq:double_integrator} uniformly distributed from the feasible region are taken, and the approximate controller is then trained using L-BFGS for different hidden layers and neurons per hidden layer. The resulting worst case approximation error $\gamma$ can be seen in Figure~\ref{fig:approximation_error_1000}. Note that for sufficiently large NN architectures $\gamma < \zeta_{180}$, but $\gamma > \zeta_{1.1}$ for all architectures. Hence, with the approach from \cite{fabiani2021} it is possible to show ISS, but not exponential stability. In Figure~\ref{fig:approximation_error_lipschitz_constant_1000} we also plot the Lipschitz constant $\mathcal{L}_{\infty}\left(\psi_1(x)-\psi_2(x), \mathbb{X}_{N}\right)$ of the approximation error. Note that it is generally orders of magnitude larger than $\vartheta$, making it impossible to verify asymptotic stability even if the approximation error is small enough.

\begin{figure}
    \centering
    \includegraphics[width=0.48\textwidth]{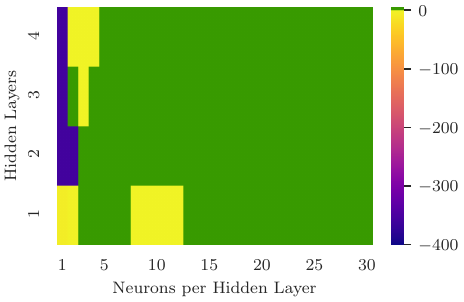}
    \caption{Stability certificate $\xi$ for solving \eqref{eq:stability_direct_opt_org}. The green area shows successful verification of asymptotic stability.}
    \label{fig:verify_stability_sufficient_1000}
\end{figure}

As a comparison, we run our direct verification method by solving \eqref{eq:stability_direct_opt_org}. The results can be seen in Figure~\ref{fig:verify_stability_sufficient_1000}. Note that we not only verify asymptotic stability, but did so even for NN policies with a very small number of hidden layers and neurons, where the method in \cite{fabiani2021} failed to show ISS.

\begin{diffadd}
The failure to verify stability for NNs with 1 hidden layer and 8–12 neurons per layer is an artifact of the stochastic learning process. The outcome of the learning process depends on the initial condition and general hyperparameters (e.g., batch size, learning rate, etc.). Hence, with a more careful training procedure, and hyperparameter tuning, it is to be expected to be able to learn a NN control policy that can be verified for stability.
\end{diffadd}

\subsection{Numerical Experiments} \label{subsec:num_exp}

In this section, we assess the numerical performance of our approach. All problems were solved using Gurobi 9.5 \cite{gurobi} using 32 Threads on a workstation with an AMD Ryzen Threadripper 3990X 4.3 GHz CPU and 32 GB of RAM.

We approximate and verify against MPC problem \eqref{original_MPC}, and use a ReLU NN architecture as introduced in Section~\ref{subsec:relu_nn} to approximate the MPC policy. For the training, we collect 2000 samples from the MPC controller uniformly distributed on the feasible region, and the NN is trained using PyTorch \cite{paszke2019} with ADAM \cite{kingma2015} as the optimizer.

Here, we are not necessarily interested in the approximation accuracy of the NN, but more in the numerical performance and scaling of the to-be-solved optimization problems. Applying this method to a specific problem requires naturally more tuning since the verification will only succeed if the NN approximates the MPC policy closely enough, but the presented results will still give a quantitative insight in solve times and what can be expected if applied to a more specific problem.

\begin{figure}
    \centering
    \includegraphics[width=0.48\textwidth]{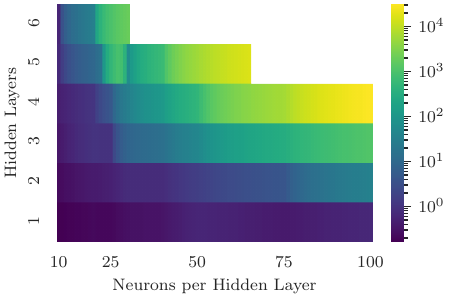}
    \caption{Solve times of \eqref{eq:error_robust_MPC_verification} in seconds with MPC horizon $N=10$.}
    \label{fig:comp_times_abs_error_10}
\end{figure}

In a first experiment, we fix the horizon of the MPC controller to be $N=10$, and solve problem \eqref{eq:error_robust_MPC_verification} for different NN architectures varying in the number of hidden layers and number of neurons per hidden layer. The resulting solve times can be seen in Figure~\ref{fig:comp_times_abs_error_10}. We can note that the computation time increases exponentially with the number of hidden layers and neurons per hidden layer. This is as expected, since each neuron adds a binary decision variable to the MILP, increasing the solve time exponentially. On the other hand, this also means that deeper architectures are preferred since the representational power of the NN grows faster with depth as compared to increasing the number of neurons per hidden layer, resulting in architectures which can be verified more quickly.
    
\begin{figure}
    \centering
    \includegraphics[width=0.48\textwidth]{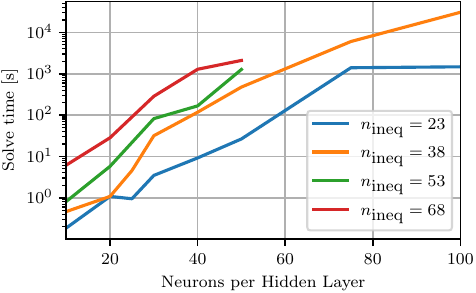}
    \caption{Solve times of \eqref{eq:error_robust_MPC_verification} in seconds for 4 hidden layers and MPC horizon lengths $N=5,10,15,20$.}
    \label{fig:comp_times_ineq}
\end{figure}

Next, we investigate the influence of the horizon length $N$ of the MPC problem. For this, we fix the NN to have four hidden layers. Increasing the horizon increases the number of inequalities in the parametric QP. Hence, we expect the solve times to increase since we add a binary decision variable per inequality, which can be observed in Figure~\ref{fig:comp_times_ineq}. The computation time increases with the number of inequalities and the number of neurons per hidden layer. Note that the complexity does not depend on the state dimension of the system, meaning that our method does not necessarily suffer from the curse of dimensionality.

\begin{figure}
    \centering
    \includegraphics[width=0.48\textwidth]{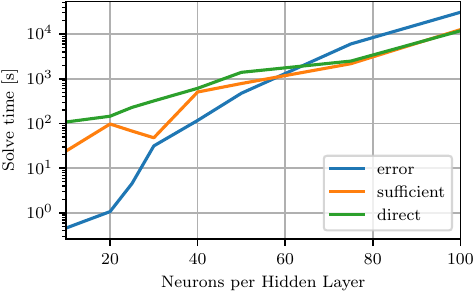}
    \caption{Solve times of \eqref{eq:error_robust_MPC_verification}, \eqref{eq:stability_opt_org}, and \eqref{eq:stability_direct_opt_org} in seconds for 4 hidden layers and MPC horizon $N=10$.}
    \label{fig:comp_times_all_methods}
\end{figure}

As a final experiment, we compare the different methods against each other. For this, we fix the horizon to $N=10$ and the NN to have four hidden layers. We solve for the maximum absolute error, and verify the MPC approximation using the sufficient and direct methods introduced in Section~\ref{sec:stability_verification} by solving problems \eqref{eq:error_robust_MPC_verification}, \eqref{eq:stability_opt_org}, and \eqref{eq:stability_direct_opt_org} respectively. The resulting computation times can be seen in Figure~\ref{fig:comp_times_all_methods}. We can see that for small problems, solving for the maximum absolute error between the MPC and the NN is faster than the verification methods. But for bigger problems the verification problems are solved more quickly, although solving an indefinite MIQP is considered to be more difficult. Interestingly, the direct verification method is slower to solve than the sufficient method. Potentially, the cascading structure of the direct verification problem makes the problem harder to solve.

\begin{diffadd}
\begin{remark}
    While solving the described verification problems can be very costly, the computation times are still manageable in practice. The alternative are statistical methods to verify stability. While these methods can be faster if only low probabilistic guarantees are needed, they tend to scale very badly if a high probabilistic confidence is required. For example, in \cite{hertneck2018} a problem similar in size to \eqref{eq:double_integrator} is considered. They need roughly 500 hours to verify that the approximate learned controller is stable for the closed loop system, with a probability of only 99\%.
\end{remark}
\end{diffadd}

\section{Summary}

We have introduced a flexible framework that allows one to formulate a variety of verification problems as MILPs or indefinite MIQPs. In particular, we have shown that not only NN with linear layers and ReLU activation functions are representable in such a framework, but also more complex layer structures such as parametric QPs and piecewise affine functions on polyhedral sets can be represented and combined.

Taking these formulations, we then constructed an MILP problem, which gives us the maximum approximation error between a given MPC and its NN approximation. Together with an MPC formulation robust in the input (ex. tube MPC), we can show constraint satisfaction and stability in the feasible region of the MPC controller. Alternatively, we provided indefinite MIQP formulations to directly verify stability and calculate an outer approximation of the stability region.

We compared our approach against a Lipschitz based method and showed that our approach outperforms it, being able to verify asymptotic stability in cases where the Lipschitz based method could only show ISS or could not show stability at all.

\bibliographystyle{IEEEtran}
\bibliography{refs}

\appendices

\section{Improved Bounds for Parametric QP Formulations} \label{sec:bound_improvementes}

As we have seen in Section~\ref{subsec:parametric_qps}, the complementarity constraints of the KKT conditions can be linearized via ``big-M" constraints \eqref{eq:KKT_complementarity_big_M}. In practice, choosing the constant $M$ too large can lead to numerical instability. Hence, ideally, one should calculate a tight bound beforehand \cite{bemporad1999}. Also, unless the parametric QP forms the last layer, then bounds on the output $y$ may even be required by the next layer.

Assuming we are given or have calculated upper and lower bounds on the input, $\underline{x} \leq x \leq \overline{x}$, we can calculate
\begin{equation} \label{eq:bound_MILP_primal}
\begin{aligned}
    M^z_i = \max_{z, \lambda, \mu_\text{eq}} \quad & g_i - F_iz \\
    \text{s.t.}~ \quad & \underline{x} \leq Dz \leq \overline{x} \\
    & Az = b, Fz \leq g \\
    & Pz + q + A^T\mu_\text{eq} + F^T\lambda = 0 \\
    & \lambda \geq 0
\end{aligned}
\end{equation}
for $i = 1, \dots, n_\text{ineq}$ and
\begin{equation} \label{eq:bound_MILP_lambda}
\begin{aligned}
    M^\lambda_i = \max_{z, \lambda, \mu_\text{eq}} \quad & \lambda_i \\
    \text{s.t.} ~\quad & \underline{x} \leq Dz \leq \overline{x} \\
    & Az = b, Fz \leq g \\
    & Pz + q + A^T\mu_\text{eq} + F^T\lambda = 0 \\
    & \lambda \geq 0
\end{aligned}
\end{equation}
for $i = 1, \dots, n_\text{ineq}$. Then, we can set
\begin{equation*}
    M = \max(\max_i M^z_i, \max_i M^\lambda_i).
\end{equation*}
If the linear programs \eqref{eq:bound_MILP_primal} or \eqref{eq:bound_MILP_lambda} are unbounded, one may restrict their feasible sets by including the MILP constraints
\begin{subequations}
\begin{align}
    & 0 \leq \lambda_i \leq M\beta_i \\
    & 0 \leq g_i - F_iz \leq M(1-\beta_i).
\end{align}
\end{subequations}
Even though the resulting MILPs can be solved orders of magnitude faster than the actual verification problem, their solution may still take a considerable amount of time because we are solving $2n_\text{ineq}$ problem instances. In practice, these MILPs are not solved to optimality but are interrupted after a predefined timeout, and the best LP-Relaxation bound is used.

\section{Technical Proofs}
\label{subsec:proof_lem:inf_norm_MILP}

\begin{proof}[Proof of Lemma~\ref{lem:inf_norm_MILP}]
As the polytopic set $\mathbb{X}$ is bounded, there exists a bounded box $[\underline{t},\overline{t}] \supseteq \mathbb{X}$. For every $i=1,\dots,m$, introduce a continuous decision variable $z_i \geq 0$ as well as a binary decision variable $\beta_i\in \{0,1\}$ that satisfy
\begin{subequations} \label{eq:norm_const}
\begin{align}
    t_i &\leq z_i \leq t_i + 2\overline{t}\beta_i, \label{eq:norm_const_a} \\
    -t_i &\leq z_i \leq -t_i -2\underline{t}(1-\beta_i). \label{eq:norm_const_b}
\end{align}
\end{subequations}
If $t_i > 0$, then \eqref{eq:norm_const_b} implies that $\beta_i=0$, and \eqref{eq:norm_const_a} implies that $z_i=t_i$. If $t_i < 0$, on the other hand, then \eqref{eq:norm_const_a} implies that $\beta_i=1$, and \eqref{eq:norm_const_b} implies that $z_i=-t_i$. If $t_i=0$, finally, then \eqref{eq:norm_const} implies that $z_i=0$ irrespective of $\beta_i$. In any case, we thus have $z_i=|t_i|$.
In order to express the infinity norm $f(t)=\max(z_1,\dots,z_m)$ in terms of linear constraints, we introduce continuous decision variables $\tau$ and $\tilde{z}_i$ for every $i=1,\dots,m$ as well as the binary decision variables $\tilde{\beta}_i \in \{0,1\}$, $i=1,\dots,m$, subject to the constraints
\begin{subequations}
\begin{align}
    1 &= \sum_{i=1}^m \tilde{\beta}_i, \label{eq:max_ineq_a} \\
    \tau &= \sum_{i=1}^m \tilde{z}_i, \label{eq:max_ineq_b} \\
    \tau &\geq z_i \hspace{28.6mm} \forall i=1,\dots,m, \label{eq:max_ineq_c} \\
    0 &\leq \tilde{z}_i \leq \max(|\underline{t}|,|\overline{t}|)\tilde{\beta}_i \quad \forall i=1,\dots,m, \label{eq:max_ineq_d} \\
    \tilde{z}_i &\leq z_i \hspace{28.6mm} \forall i=1,\dots,m. \label{eq:max_ineq_e}
\end{align}
\end{subequations}
The forcing constraints \eqref{eq:max_ineq_d} ensure that $\tilde{z}_i = 0$ whenever $\tilde{\beta}_i=0$, whereas \eqref{eq:max_ineq_a} ensures that $\tilde{\beta}_i=1$ for exactly one $i$. Hence, \eqref{eq:max_ineq_b} sets $\tau = \tilde{z}_i$ for the unique $i$ with $\tilde{\beta}_i=1$. By \eqref{eq:max_ineq_e}, we thus have $\tau = \tilde{z}_i \leq z_i \leq \max(z_1,\dots,z_m)$. The constraint \eqref{eq:max_ineq_c}, finally, ensure the converse inequality $\tau \geq z_i \geq \max(z_1,\dots,z_m)$. Hence, $f(t)=\|t\|_\infty$ is MILP-representable, and its graph can be expressed as
\begin{equation*}
    \gr(f) = 
    \left\{
    (t,\tau) \;\middle|\; \begin{aligned}
    \exists \beta, \tilde{\beta} &\in \{0,1\}^m, \; z, \tilde{z} \in \mathbb{R}^m: \\ \forall i& = 1,\dots,m: \\
    t_{\phantom{i}} &\in \mathbb{X}, \\
    t_i &\leq z_i \leq t_i + 2\overline{t}\beta_i, \\
    -t_i &\leq z_i \leq -t_i -2\underline{t}(1-\beta_i), \\
    1 &= \sum_{j=1}^m \tilde{\beta}_j, ~\tau = \sum_{j=1}^m \tilde{z}_j, \\
    0 &\leq \tilde{z}_i \leq \max(|\underline{t}|,|\overline{t}|)\tilde{\beta}_i, \\
    \tilde{z}_i &\leq z_i \leq \tau
    \end{aligned}
    \right\}.
\end{equation*}
This observation completes the proof.
\end{proof}

\subsection{Proof of Corollary~\ref{cor:1_norm}} \label{subsec:proof_cor:1_norm}

\begin{proof}
The proof widely parallels that of Lemma~\ref{lem:inf_norm_MILP}. However, there is no need to introduce decision variables $\tilde{\beta}$ and $\tilde{z}$ because we can directly set
\begin{equation*}
    \tau = \|t\|_1 = \sum_{i=1}^m z_i.
\end{equation*}
Thus, the claim follows.
\end{proof}

\end{document}